\documentclass{elsarticle}
\usepackage{cmll}
\usepackage{latexsym}
\usepackage{amssymb,amsmath,amsbsy}
\usepackage{graphicx}
\usepackage{stmaryrd}
\usepackage{pstricks}
\usepackage[all]{xy}
\usepackage{enumerate}
\usepackage[colorlinks=true, pdfstartview=FitV, linkcolor=blue, citecolor=blue, urlcolor=blue]{hyperref}
\usepackage{mathrsfs}
\usepackage{amsthm}
\usepackage{qtree}
\usepackage{verbatim}
\usepackage{framed, color}
\usepackage{proof}

\newtheorem{definition}{Definition}
\newtheorem{theorem}{Theorem}
\newtheorem{property}{Property}
\newtheorem{example}{Example}
\newtheorem{remark}{Remark}
\newtheorem{lemma}{Lemma}
\newtheorem{corollary}{Corollary}

\newtheorem{notation}{Notation}{\bf}{\it}
\newcommand{\lin}{\mathrel{-\!\circ}}
\newcommand{\dem}{\small \triangleright \normalsize}

\newcommand{\der}{\vdash }
\newcommand{\deri}{\vdash_{\tt I} }
\newcommand{\ders}{\vdash_{\tiny \STA} }

\newcommand{\STA}{{\tt STA}}

\newcommand{\STR}{{\tt STR}}
\newcommand{\IN}{{\tt INTER}}
\newcommand{\DLAL}{{\tt DLAL}}
\newcommand{\LAL}{{\tt LAL}}

\newcommand{\FPTIME}{{\tt FPTIME}}

\newcommand{\FV}[1]{{\tt FV}(#1)}
\newcommand{\FTV}[1]{{\tt FTV}(#1)}
\newcommand{\dom}[1]{{\tt dom}(#1)}
\newcommand{\rk}[1]{{\tt rk}(#1)}
\newcommand{\rks}[1]{{\tt rks}(#1)}
\newcommand{\D}[1]{{\tt d}(#1)}
\newcommand{\DS}[1]{{\tt ds}(#1)}
\newcommand{\W}[2]{{\tt W}(#1, #2)}
\newcommand{\WS}[2]{{\tt Ws}(#1, #2)}

\newcommand{\SN}{\bf \rm SN}

\newcommand{\sm}[3]{#1_{#2}^{#3}}

\newcommand{\eqdom}{\boxminus}
\newcommand{\lam}{\lambda}
\newcommand{\M}{{\tt M}}
\newcommand{\N}{{\tt N}}
\newcommand{\R}{{\tt P}}
\newcommand{\Q}{{\tt Q}}
\newcommand{\U}{{\tt U}}
\newcommand{\V}{{\tt V}}
\newcommand{\x}{{\tt x}}
\newcommand{\y}{{\tt y}}
\newcommand{\z}{{\tt z}}
\newcommand{\w}{{\tt w}}

\newcommand{\s}[1]{{\tt s}_{#1}}
\newcommand{\redbeta}{\longrightarrow_{\beta}}
\newcommand{\redbetas}{\overset{*}{\longrightarrow_{\beta}}}

\newcommand{\sub}[3]{#1 [#2/#3]}
\newcommand{\subi}[5]{#1 [#2/#3,...,#4/#5]}

\newcommand{\T}{{\mathcal T}}
\newcommand{\TS}{{\mathcal TS}}
\newcommand{\A}{{\tt A}}
\newcommand{\B}{{\tt B}}
\newcommand{\C}{{\tt C}}
\newcommand{\tvar}{{\tt a}}
\newcommand{\tvarb}{{\tt b}}

\renewcommand{\bar}[1]{\overline{#1}}
\newcommand{\stra}[1]{\{ #1 \}}
\newcommand{\str}[2]{\{ #1_{1}, ...\ , #1_{#2} \}}

\newcommand{\num}[1]{\underline{ #1 }}
\newcommand{\Num}{{\mathcal N}}

\newcommand{\Wtype}{{\bf W}}
\newcommand{\Witype}{{\bf WI}}
\newcommand{\Wtypei}[2]{{\bf W}_{#1,#2}}
\newcommand{\Witypei}[2]{{\bf WI}_{#1,#2}}
\newcommand{\SI}[1]{(#1)^{\circ}}

\begin{document}
\title{A type assignment for $\lambda$-calculus complete both for FPTIME and strong normalization}
\author{Erika De Benedetti and Simona Ronchi Della Rocca}
\date{}
\address{Dipartimento di Informatica - Corso Svizzera 185 - 10149 Torino\\
{\tt \{debenede, ronchi\}@di.unito.it}}

\begin{abstract}
One of the aims of Implicit Computational Complexity is the design of programming languages with bounded computational complexity; indeed, guaranteeing and certifying a limited resources usage is of central importance for various aspects of computer science. One of the more promising approaches to this aim is based on the use of lambda-calculus as paradigmatic programming language and the design of type assignment systems for lambda-terms, where types guarantee both the functional correctness and the complexity bound. Some systems characterizing polynomial time complexity have been designed, inspired by the Light Logics, all of which give type to a proper subset of strongly normalizing terms.\\
We propose a system of stratified types, inspired by intersection types, where intersection is a non-associative operator. The system, called \STR, is correct and complete for polynomial time computations; moreover, all the strongly normalizing terms are typed in it, thus increasing the typing power with respect to the previous proposals.
Finally, \STR\ enjoys a stronger expressivity with respect to the previous system \STA, since it allows to type a restricted version of iteration.
\end{abstract}

\maketitle
\section{Introduction}
The importance of controlling (and producing a formal certification of) the resource usage of programs is already recognized by the scientific community. In this general setting, we are interested in the design of programming languages with an intrinsically polynomial computational bound. 
We are interested in an ML-like approach, so our starting points will be:
\begin{itemize}
\item the use of  $\lambda$-calculus as an abstract paradigm of programming languages;
\item the use of types to certificate program properties.
\end{itemize}
In this line, the aim is to design a type assignment system for $\lambda$-calculus, where
types certificate both the functional correctness and the polynomial bound of terms.
There are already two proposals along this line: the systems \DLAL\ by Baillot and Terui \cite{Baillot04lics} and the system \STA\ by Gaboardi and Ronchi Della Rocca \cite{GaboardiRonchi07csl}. Both systems are based on Light Logics, derived from the Linear Logic of Girard \cite{Girard87tcs}. More precisely, types of \DLAL\ are a proper subset of formulae of \LAL\ by Asperti and Roversi \cite{Asperti02tocl}, a simplified affine version of the Light Linear Logic of Girard \cite{Girard98ic}, while types of \STA\ are a proper subset of formulae of the Soft Linear Logic by Lafont \cite{Lafont02SLL}. The design of both systems is based on the transfer of the complexity properties from logics to terms, according to the proofs-as-programs approach inspired by the Curry-Howard isomorphism. Both characterize the polynomial functions, in the sense that all and only the polynomial functions can be coded in such systems, according to the standard coding of functions by $\lambda$-terms. The logical inspiration of both systems is at the same time the key ingredient of their correctness and the responsible for their weak expressivity, since they can code few algorithms. From a typability point of view,
both systems give types to a proper subset of the strongly normalizing terms.\\
A stronger expressive power could be achieved by enriching the language, and this approach has been followed by many authors.
In particular, an extension of \STA\ has been designed in \cite{ChrzaszczShubert/csl/2012}, where some features like ML-polymorphism have been added to $\lambda$-terms, and a typed extension of \DLAL\ has been proposed in \cite{BaillotGaboardiMogbil10esop}, where a typed recursion has been introduced, besides other programming features.

Here we want to explore a different direction; namely, we want to preserve having pure $\lambda$-calculus as a programming language, but at the same time we want to design a system with stronger typability power, hoping to obtain, as side effect, also a gain in expressivity.
% , and we want to extend the set of types in order to assign type to more terms, while preserving the polynomial bound.
The resulting system, called \STR\, is polynomial (in the previous sense), and moreover all the strongly normalizing terms are typed in it, so increasing in a very significant way the typability power with respect to both \DLAL\ and \STA. In particular, \STR\ is more expressive than \STA, since a restricted form of iteration can be typed in \STR, which cannot be expressed in \STA.\\
In \STR\ types can be either linear or stratified. Linear types represent linear premises, in the sense of Linear Logic, and the operation of stratification is a sort of soft promotion. From a logical point of view, while in \STA\ the promotion is a sort of multiple contraction for different copies of the same premise, here we can contract also premises having different types. This feature can no more be expressed in a logical way; indeed, \STR\ is introduced as type assignment system without a pure logical counterpart.\\
In order to build \STR\  we were inspired by intersection types \cite{CoppoDezani:NDJoFL-80}. Indeed, the relation between \STA\ and \STR\ reminds, in a very rough way, the relation between simple types and intersection types assignment system, the second being derived from the first, but allowing a variable to be assigned different types. The relation of the present work with intersection types is further discussed in the conclusion.\\
\STR\ preserves the polynomial bound; indeed, the introduction of the intersection increases the typability power without increasing the computability power, as proved in \cite{BucciaPiper03}.
In the relation between \STA\ and \STR, the same phenomenon happens: \STR\ allows to type all the strongly normalizing terms, but it characterizes the same class of functions as \STA, i.e. \FPTIME.
However the expressivity is increased, since bounded iteration cannot be expressed in \STA\ , while in \STR\ a restricted iteration construct can be typed.
%In order to prove this claim, we were thinking of intersection types and of the proof, given in \cite{BucciaPiper03}, that intersection types code the same set of functions as the simple types, while increasing the typability power. In the same way, we show that both \STA\ and \STR\ characterize all and only the polynomial functions, whereas the typing power of \STR\ is higher: indeed, in \STR\ all and only the strong normalizing terms can be typed. 
%In particular, bounded iteration cannot be expressed in \STA\ , while in \STR\ a restricted iteration construct can be typed.

The paper is organized as follows:
in Section~\ref{sec:STR} we introduce the type assignment system \STR\ and we prove that it enjoys the subject reduction property.
In Section~\ref{sec:sn} we prove that \STR\ characterizes strong normalizazion.
%, by proving that any typed term is strongly normalizing (Subsection~\ref{subsec:tsn}) and that any strongly normalizing term can be assigned a type (Subsection~\ref{subsec:snt}).
In Section~\ref{sec:poly}, we
%recall system \STA\ and give the translation from \STA\ to \STR\ (Subsection~\ref{subsec:trad}); then 
we prove that \STR\ is sound and complete with respect to \FPTIME.
%, by showing that in \STR\ all and only the functions typed in \STA\ are represented and that the polynomial bound holds for the numerical functions (Subsection~\ref{subsec:W}).
%In this section we give also an example of the gain in expressivity of \STR\  with respect to \STA.
In Section~\ref{sec:strat-int}, we comment on the choice of stratified types with respect to intersection types.
Finally, in Section~\ref{sec:concl}, we conclude with some technical observations on the use of intersection types for quantitative purposes.
\label{sec:intro}
\section{The \STR\ type assignment system}
\label{sec:STR}
In this section we introduce the type assignment system for $\lambda$-calculus named \STR, based on the notion of \textsl{stratification} of types, and we prove that it enjoys subject reduction.

\begin{definition}\label{def:int-types}
\begin{enumerate}[(1)]\ 
\item The set $\Lambda$ of {\em terms} is defined by the following syntax:
$$\M ::= \x \mid \lam \x. \M \mid \M \M$$
where $\x$ ranges over a countable set ot variables $\mathcal{X}$. 
$\FV{\M}$ denotes the set of free variables of the term $\M$.
Terms are considered modulo $\alpha$-equivalence; moreover, bound variables are assumed to be all distinct and different from free ones.
The symbol $\equiv$ denotes the identity on terms, modulo renaming of bound variables.\\
A {\em (term) context} is generated by the same grammar, starting from a constant $[.]$ (the hole), in addition to variables.
Term contexts are denoted by $\C[.]$, and $\C[\M]$ denotes the result of plugging $\M$ into every occurrence of $[.]$ in $\C[.]$. Observe that, as usual, the plugging operation allows the capture of free variables.

\item The reduction relation $\redbeta$ is the contextual closure of the rule ${(\lam \x. \M) \N \rightarrow \M[\N/\x]}$, where the substitution ${\M[\N_{1}/\x_{1}, ... , \N_{n}/\x_{n}]}$, also denoted by $\M[\N_{i}/\x_{i}]_{i=1}^{n}$, is the capture-free substitution of $\N_{i}$ to all the free occurrences of $\x_{i}$ in $\M$ ${(1 \leq i \leq n)}$. The relation $\redbetas$ is the reflexive and transitive closure of $\redbeta$.

\item The set of pre-types is defined by the following syntax:
$$\A::= \tvar \mid \sigma \lin \A \mid \forall \tvar.\A \quad \mbox{ (linear pre-types)}$$
$$\sigma ::= \A \mid  \stra{\underbrace{\sigma,...,\sigma}_{n}} \quad  \mbox{ for } n>0 
\mbox{ (stratified pre-types) }$$
where $\tvar$ ranges over a countable set of type variables. Type variables are ranged over by $\tvar,\tvarb$, linear pre-types are ranged over by $\A, \B, \C$, and stratified pre-types by 
$\sigma, \tau$. $\FTV{\sigma}$ denotes the set of free type variables of $\sigma$.

Let $\sim$ denote the syntactical equality between (stratified) pre-types. On pre-types we define the following equivalence $=$, modulo renaming of bound variables:

$\begin{array}{lcl}
\A \sim \B &\mbox{ implies }& \A=\B\\ 
\A = \B & \mbox{ implies } & \forall \tvar. \A = \forall \tvar. \B \\
\sigma = \tau, \A=\B  &\mbox{ implies }& \sigma \lin \A = \tau \lin \B \\
\stra{\sigma_{1},...,\sigma_{n}}= \stra{\tau_{1},...,\tau_{m}} &\mbox{ iff } &
\forall i.\exists j. \sigma_{i}=\tau_{j} \mbox{ and }
\forall j. \exists i. \sigma_{i}=\tau_{j} \\
 (1 \leq i \leq n, 1 \leq j\leq m ) &&
\end{array}$
i.e., a stratified pre-type represents a set.

\item {\em Types} are pre-types modulo the equivalence relation $=$. The set of types is denoted by $\T$. In order to avoid reasoning modulo $=$, when writing $\{\sigma_{1}, ... , \sigma_{n}\}$ we assume that $\sigma_{i}\not = \sigma_{j}$, for $i\not=j$ ($1\leq i,j \leq n$), where $\sigma_{1}, ... , \sigma_{n}$ are the {\em components} of $\{\sigma_{1}, ... , \sigma_{n}\}$. \\
A multiset over $\T$ is an unordered list $[\sigma_{1}, ... , \sigma_{n}]$, where the number of occurrences of $\sigma_i$ is its {\em multiplicity}. The multiset union $\uplus$ is the concatenation of lists.
The multiset of the linear components of $\sigma$ is defined inductively as
\[ { \bar \A = [ \A ]} \qquad {\overline{\str{\sigma}{k}} = \bar \sigma_{1} \uplus ... \uplus \bar \sigma_{k}}. \]
%where $\uplus$ is the union operator on multisets.
We use $\{ \sigma \}^{n}$ as a short for $\underbrace{ \{ ... \{ }_{n} \sigma \underbrace{ \} ... \} }_{n}$.

\item Contexts are partial functions with finite domain from variables to types.

\end{enumerate}

\end{definition}
\begin{example} Let $\sigma = \stra{ \A, \stra{\A, \B}}$: then $\bar \sigma = [ \A, \A, \B ]$.
\end{example}

We introduce a few notations that are used throughout the paper.

\begin{notation}

{\em(Types){\bf.}} Operations on sets are naturally extended to stratified types; in particular, we denote by $\sm{\cup}{i=1}{n} \stra{\sigma_{i}}$ the stratified type obtained by unifying the singletons $\stra{\sigma_{1}}, ... , \stra{\sigma_{n}}$. In order to avoid unnecessary parentheses, we assume that $\lin$ takes precedence over $\forall$, i.e. ${\forall \tvar. \sigma \lin \A}$ is equivalent to ${\forall \tvar. (\sigma \lin \A)}$.

\medskip

{\em(Contexts){\bf.}} The domain of $\Gamma$ is denoted by $\dom{\Gamma}$;
$\stra{\Gamma}^{n}$ is the context such that $\stra{\Gamma}^{n}(\x) = \stra{\Gamma(\x)}^{n}$,
while $\emptyset$ denotes the context with empty domain.\\
%We denote by $\sm{\#}{i=1}{n} \Delta_{i}$ the fact that $j \not= h$ implies ${\dom{\Delta_{j}} \cap \dom{\Delta_{h}} = \emptyset}$ ${(1 \leq j, h \leq n)}$, while $\sm{\eqdom}{i=1}{n} \Gamma_{i}$ denotes the fact that $j \not= h$ implies ${\dom{\Delta_{j}} = \dom{\Delta_{h}}}$ ${(1 \leq j, h \leq n)}$;
We denote by ${\Gamma_{1} \# ... \# \Gamma_{n} }$ (resp. ${\Gamma_{1} \eqdom ... \eqdom \Gamma_{n}}$) the fact that $j \not= h$ implies ${\dom{\Delta_{j}} \cap \dom{\Delta_{h}} = \emptyset}$ (resp. ${\dom{\Delta_{j}} = \dom{\Delta_{h}}}$), for ${1 \leq j, h \leq n}$, and we often write it as $\sm{\#}{i=1}{n} \Delta_{i}$ (resp. $\sm{\eqdom}{i=1}{n} \Gamma_{i}$).\\
If $\sm{\#}{i=1}{n} \Gamma_{i}$, then $\Gamma_{1}, ... , \Gamma_{n}$ is the context such that ${(\Gamma_{1}, ... , \Gamma_{n})(\x)= \Gamma_{i}(\x)}$, where ${\x \in \dom{\Gamma_{i}}}$ ${(1 \leq i \leq n)}$.
If $\sm{\eqdom}{i=1}{n} \Gamma_{i}$, then ${\sm{\cup}{i=1}{n} \stra{\Gamma_{i}}}$ is the context such that $(\sm{\cup}{i=1}{n} \stra{\Gamma_{i}})(\x)= \sm{\cup}{i=1}{n} 
\stra{\Gamma_{i}(\x)}$.

\end{notation}

 The system \STR\ proves judgments of the kind $\Gamma \der \M: \A$, where $\Gamma$ is a context, $\M$ a term and $\sigma$ a type.
 The rules of the system are shown in Table~\ref{tab:STR}. 
\begin{notation}[Derivations]
Type derivations are denoted by $\Sigma$, $\Pi$.
We denote by $\Gamma \der \M: \sigma$ the existence of a derivation proving such statement, while $\Pi \dem \Gamma \der \M: \sigma$ denotes a particular derivation $\Pi$, and we abbreviate $\emptyset \der \M:\sigma$ by $\der \M:\sigma$.
Given the application of a rule in a derivation, the derivations to which is is applied are its {\em premises}.
Moreover $\dom{\Sigma}$ represents the set of term variables $\sm{\cup}{\Gamma \in \Sigma}{\;} \dom{\Gamma}$, where $\Gamma \in \Sigma$ means that $\Gamma$ is a context occurring in any application of rule of $\Sigma$; by abuse of notation, we denote by ${\Sigma_{1} \# ... \# \Sigma_{n} }$, or $\sm{\#}{i=1}{n} \Sigma_{i}$, the fact that $j \not= h$ implies ${\dom{\Sigma_{j}} \cap \dom{\Sigma_{h}} = \emptyset}$, for $1 \leq j,h \leq n$.
\end{notation}

\begin{table}
\begin{center}
\fbox{
\begin{minipage}{12cm}
\medskip
$$
\infer[(Ax)]{\x:\A \der \x:\A}{}
$$
$$
\infer[(w)]{\Gamma, \x: \A \der \M: \sigma}{\Gamma \der \M:\sigma & (\x \not\in \dom{\Gamma}) }
\qquad
\infer[(\lin I)]{\Gamma \der \lambda \x.\M:\sigma\lin \B}{\Gamma, \x: \sigma \der \M:\B }
$$
$$
\infer[(\lin E)]{\Gamma_{1},\Gamma_{2} \der \M\N:\A}
{\Gamma_{1}\der \M:\sigma\lin \A & 
\Gamma_{2} \der \N:\sigma & (\Gamma_{1} \# \Gamma_{2})}
$$
$$
\infer[(m)]{\Gamma, \x: \sm{\cup}{i=1}{n} \stra{\sigma_{i}}  \der \M[\x/\x_{1},...,\x/\x_{n}]:\tau}
{\Gamma, \x_{1}:\sigma_{1},...,\x_{n}:\sigma_{n} \der \M:\tau}
$$
$$
\infer[(st)]{\sm{\cup}{i=1}{n} \stra{\Gamma_{i}} \der \M: \stra{\sigma_{1},...,\sigma_{n}}}
{( \Gamma_{i} \der \M:\sigma_{i} )_{1 \leq i \leq n} & \sm{\eqdom}{i=1}{n} \Gamma_{i}}
$$
$$
\infer[(\forall I)]{\Gamma \der \M: \forall \tvar. \A}
{\Gamma \der \M:\A & (\tvar \not \in \FTV{\Gamma})}
\qquad
\infer[(\forall E)]{\Gamma \der\M:\sub{\B}{\A}{\tvar}}
{\Gamma\der \M:\forall \tvar .\B }
$$

\end{minipage}
}
\caption{The \STR\ Type Assignment system. \label{tab:STR}}
\end{center}
\end{table}

A few comments about the system are in order. Observe that rules $(Ax)$ and $(w)$ introduce only linear types; moreover, by rule $(\lin E)$ only terms having disjoint sets of free variables can be applied to each other: however, more general applications can be built by applying the multiplexor rule $(m)$ and renaming term variables. 
Observe that rule $(st)$ introduces the stratification both in the premises and in the subject.
Finally, rule $(\forall E)$ allows to replace type variables by linear types only, in order to preserve the syntax.

Note that a more general weakening rule is derivable:
\begin{property}\label{prop:weak}
$\Gamma \der \M:\sigma$ and $\x \not\in \dom{\Gamma}$ imply $\Gamma, \x:\tau \der \M:\sigma$, for every $\tau$.
\end{property}
\begin{proof}
By induction on $\tau$.
If $\tau$ is linear, it is sufficient to apply rule $(w)$.
Otherwise, if $\bar{\tau} = [\A_{1}, ... , \A_{n}]$ then we can build the following derivation:
\[\infer=[\delta]{\Gamma, \x: \tau \der \M: \sigma}{\infer=[(w)]{\Gamma, \x_{1}: \A_{1}, ... , \x_{n}: \A_{n} \der \M: \sigma}{\Gamma \der \M: \sigma}}\]
where $\delta$ is a suitable sequence of applications of rule $(m)$.
\end{proof}
\begin{example} Let $\Gamma \der \M: \sigma$ and ${\tau = \stra{\stra{\A}, \stra{\A, \B}} }$, so ${ \bar{\tau} = [\A, \A, \B] }$: then we can build the following derivation:
\[\infer[(m)]{\Gamma, \x: \stra{\stra{\A}, \stra{\A, \B}} \der \M: \sigma}{
\infer[(m)]{\Gamma, \x_{1}: \stra{\A}, \x_{2}: \stra{\A, \B} \der \M: \sigma}{
\infer[(m)]{\Gamma, \x_{1}: \stra{\A}, \y_{2}: \A, \y_{3}: \B \der \M: \sigma}{
\infer=[(w)]{\Gamma, \y_{1}: \A, \y_{2}: \A, \y_{3}: \B \der \M: \sigma}{\Gamma \der \M: \sigma}}}}\]
\end{example}

Note that, based on the previous property, the condition on the contexts in rule $(st)$ is not restrictive.

Rules can be classified into  \textit{constructive rules}, $\{(Ax),(\lin I), (\lin E)\}$ which contribute to building the subject,
and non-constructive ones. The latter can be further classified into \textit{quantifier rules}, $\{(\forall I), (\forall E)\}$, modifying the types but not the terms,
\textit{renaming rules}, $\{(w),(m)\}$,
renaming variables in terms or introducing new variables in the context, and the \textit{stratification rule}, $(st)$, merging derivations having the same subject. A sequence of applications of renaming (and quantifier) rules is called a {\em renaming (and quantifier) sequence}.

\begin{definition}[Instance of a term] \label{def:instance} 
A term $\M$ is an {\em instance} of $\M'$ if 
there are
an integer $n \geq 0$, $\mathcal{X}_1,\ldots,\mathcal{X}_n$ subsets of $\FV{\M'}$ and
fresh variables $\y_1,\dots, \y_n$ such that
$\M$
is obtained from $\M'$ by renaming all variables in $\mathcal{X}_i$ by $\y_i$. 
An instance is a {\em copy} if 
all the sets $\mathcal{X}_i$ are singletons. 
The notions of copy can be extended to contexts and derivations in a straightforward way.
\end{definition}
In order to easily reason about derivations in proofs, we introduce the notion of clean derivation.
The proof of the following property is obvious, thanks to the renaming rule $(m)$.
\begin{property}
A derivation is {\em clean} if, in every application of rule $(\lin E)$ with premises $\Pi_{1}$ and $\Pi_{2}$, $\Pi_{1} \# \Pi_{2}$. 
For every derivation $\Pi$, there is a clean derivation $\Pi'$ proving the same statement.
\end{property}
From now on, we assume that all derivations are clean.

\begin{lemma}\label{rem:instance} Given $\Pi\dem\Gamma \der \M: \sigma$ and a instance $\N$ of $\M$, there is $\Delta$ such that $\Sigma \dem \Delta \der \N: \sigma$ where $\Sigma$ is obtained from $\Pi$ by applying a suitable sequence of $(m)$ rule.
\end{lemma}

We also introduce some notations for renaming rules.

\begin{definition}
The {\em domain} and the {\em range} of an application of $(m)$ rule are respectively the set of variables contracted by it and the singleton of the new introduced variable. The {\em domain} and {\em range} of an application $(w)$ rule are respectively the empty set and the singleton of the new introduced variable. \\
Two applications of renaming rules are {\em disjoint} iff both their domains and their ranges are disjoint.
%Let $\Pi\dem\Gamma, \x:\sigma \der \M: \tau$. The set of ancestors of $\x$ in $\Pi$ (denoted by $A(\x,\Pi)$) is defined inductively as follows:
%\begin{itemize}
%\item if $\Pi$ is $$\infer[(Ax)]{\x:A \der \x:A}{}$$ then $A(\x,\Pi )=\{\x\}$.
%\item if $\Pi$ is $$\infer[(w)]{\Gamma, \y:A \der \M:\sigma}{\Pi'\dem \Gamma\der \M:\sigma}$$ then if $\y=\x$ then $A(\x,\Pi )=\{\x\}$, otherwise $A(\x,\Pi )=A(\x,\Pi' )$.
%\item if $\Pi$ is $$\infer[(\lin I)]{\Gamma \der \lambda \y.\M:\sigma\lin \B}{\Pi'\dem\Gamma, \y: \sigma \der \M:\B }
%$$
%then $A(\x,\Pi )=A(\x,\Pi' )$.
%\item if $\Pi$ is $$\infer[(m)]{\Gamma, \x: \str{\sigma}{n} \der \subi{\M}{\x}{\x_1}{\x}{\x_n}: \tau}
%{\Pi'\dem\Gamma, \x_1:\sigma_1,...,\x_n:\sigma_n\der \M: \tau}
%$$
%then $A(\x,\Pi)=\{\x\}\cup_{1\leq i\leq n} A(\x_i,\Pi')$. 
%\item if $\Pi$ is
%$$\infer[(st)]{\sm{\cup}{i=1}{n} \stra{\Gamma_{i}} \der \M: \stra{\sigma_{1},...,\sigma_{n}}}
%{(\Pi_i \dem \Gamma_{i} \der \M:\sigma_{i})_{1 \leq i \leq n} }$$
%then $A(\x,\Pi)=\cup_{1\leq i\leq n}A(\x,\Pi_i)$.
%\item in all other cases, $\Pi$ ends by a rule:
%$$
%\infer[(R)]{\Gamma \der \M: \sigma}{\Pi' \dem\Gamma \der \M':\sigma'}$$
%and $A(\x,\Pi)=A(\x,\Pi')$.
%\end{itemize} 
\end{definition}

The following is a key property of the system.
%allows to deconstruct a derivation assigning a stratified type into a given number of premises assigning a simpler type:

\begin{property}[Subject with stratified type]
\label{prop:stra}
Let $\Pi \dem \Gamma \der \M: \str{\sigma}{n}$; then there are $\Pi_{i} \dem \Gamma_{i} \der \N: \sigma_{i}$ $(1 \leq i \leq n)$ such that $\Pi$ consists of an application of rule $(st)$ with premises $(\Pi_{i})_{1 \leq i \leq n}$, followed by a renaming sequence.
\end{property}
\begin{proof}
By induction on $\Pi \dem \Gamma \der \M: \str{\sigma}{n}$; observe that the last rule of $\Pi$ can be either $(w)$, $(m)$ or $(st)$.
If $\Pi$ ends with an application of rule $(w)$ or $(m)$, then the proof follows by induction. Otherwise, let $\Pi$ be
\[ \infer[(st)]{\sm{\cup}{i=1}{n} \stra{\Gamma_{i}} \der \M: \str{\sigma}{n} }{(\Gamma_{i} \der \M: \sigma_{i})_{1 \leq i \leq n}} \]
Then the proof is trivial, since $\M \equiv \N$ and the renaming sequence is empty.
\end{proof}

\begin{corollary}\label{cor:stra}
Let $\Gamma \der \M: \str{\sigma}{n}$ and $\bar{\sigma} = [ \A_{1}, ... , \A_{n} ]$; then there is $\Gamma_{i}$ such that $\Gamma_{i} \der \M: \A_{i}$ $(1 \leq i \leq n)$.
\end{corollary}

The Generation Lemma connects the shape of a term with its possible typings.
%Observe that, by the previous Property, we can consider w.l.o.g. derivations assigning linear types to terms, since $\M$ being an instance of $\M'$ implies that they have the same structure.

\begin{lemma}[Generation Lemma]\label{lem:gen}
Let $\Pi \dem \Gamma \der \M: \A$.
\begin{enumerate}
\item $\FV{\M}\subseteq\dom{\Gamma}$.

\item Let $\M \equiv \x$. Then $\Pi$ consists of an $(Ax)$ rule followed by a (possibly empty) renaming and quantifier sequence.

\item Let $\M \equiv \lambda \x. \N$. Then $\A = \forall \vec \tvar. \tau \lin \B$ and there is $\lambda \x.\R$ such that $\lambda \x.\N$ is an instance of $\lambda \x.\R$ and there is a derivation:
$$
\infer=[\delta]{\Gamma \der \lambda \x.\N:\forall \vec \tvar. \tau \lin \B}
{\infer[(\lin I)]{\Gamma' \der \lambda \x.\R: \rho \lin \C}{\Gamma',\x:\rho \der \R:\C}}
$$
for some $\rho,\C$, where $\delta$ is a (possibly empty) renaming and quantifier sequence.

\item Let $\M \equiv \N\R$. Then there is $\N'\R'$ such that $\N\R$ is an instance of $\N'\R'$ and there is a derivation:
$$
\infer=[\delta]{\Gamma \der \N\R:\A'}
{\infer[(\lin E)]{\Gamma_{1},\Gamma_{2} \der \N'\R': \A}{\Gamma_{1} \der \N':\sigma \lin \A & \Gamma_{2} \der \R':\sigma & \Gamma_{1}\#\Gamma_{2}}}
$$

for some $\sigma$, where $\delta$ is a (possibly empty) renaming and quantifier sequence.

\end{enumerate}
\end{lemma}

\begin{proof} 
Easy.
\end{proof}

The following technical Lemma is useful to prove the Substitution Lemma.

\begin{lemma}\label{lem:st-m}
Let $\Sigma_{i}\dem\Theta_{i} \der \N_{i} : \sigma_{i}$ be a copy of $\Pi_{i}\dem\Delta_{i} \der \N : \sigma_{i}$ such that $\sm{\#}{i=1}{n} \Theta_{i}$ and $\sm{\eqdom}{i=1}{n} \Delta_{i}$, with $\Delta = \sm{\cup}{i=1}{n} \stra{\Delta_{i}}$ $(1 \leq i \leq n)$.
Then, from every ${\Gamma, \Theta_{1}, ... , \Theta_{n} \der \M[\N_{i}/\x_{i}]_{i=1}^{n}: \tau}$ such that $\Gamma \# \Delta$, we can derive ${\Gamma, \Delta \der \M[\N/\x_{i}]_{i=1}^{n} : \tau}$ by a renaming sequence.
\end{lemma}

\begin{proof}
By induction on $\Delta$.
If $\Delta = \emptyset$, then $\M[\N_{i}/\x_{i}]_{i=1}^{n} = \M[\N/\x_{i}]_{i=1}^{n}$, so the  renaming sequence is empty.

Otherwise, let $\Delta_{i} = \Delta'_{i}, \y: \rho_{i}$ and ${\Delta = \Delta'_{i}, \y: \rho = \sm{\cup}{i=1}{n} \stra{\Delta'_{i}}, \sm{\cup}{i=1}{n} \stra{\rho_{i}}}$.
Moreover,
let $\Theta'_{i}, \y_{i}: \rho_{i} \der \N_{i}[\y_{i}/\y] : \sigma_{i}$ be a copy of ${\Delta'_{i}, \y: \rho_{i} \der \N: \sigma_{i}}$ such that $\sm{\#}{i=1}{n} \Theta'_{i}$ and $\sm{\eqdom}{i=1}{n} \Delta'_{i}$, with $\Delta' = \sm{\cup}{i=1}{n} \stra{\Delta'_{i}}$ $(1 \leq i \leq n)$.
By inductive hypothesis, from ${\Gamma, \Theta'_{1}, ... , \Theta'_{n}, \y_{1}: \rho_{1}, ... , \y_{n}: \rho_{n} \der \M[\N_{i}[\y_{i}/\y]/\x_{i}]_{i=1}^{n}: \tau}$ such that ${\Gamma \# \Delta'}$, we can derive ${\Gamma, \Delta', \y_{1}: \rho_{1}, ... , \y_{n}: \rho_{n} \der \M[\N[\y_{i}/\y]/\x_{i}]_{i=1}^{n} : \tau}$ by a renaming sequence:
from this derivation we obtain ${\Gamma, \Delta', \y: \rho \der \M[\N/\x_{i}]_{i=1}^{n} : \tau}$ by applying rule $(m)$ with domain $\{ \y_{1}, ... , \y_{n} \}$ and range $\{ \y \}$.
\end{proof}

We supply an example in order to make the previous lemma clearer.

\begin{example}
Let
$$\infer[(st)]{\x: \stra{\sigma_{1}, \sigma_{2}}, \y: \stra{\tau_{1}, \tau_{2}}, \z: \stra{\rho_{1}, \rho_{2}} \der \N: \stra{\phi_{1}, \phi_{2}}}{\Sigma_{1} \dem \x: \sigma_{1}, \y: \tau_{1}, \z: \rho_{1}\der \N: \phi_{1} & \Sigma_{2} \dem \x: \sigma_{2}, \y: \tau_{2}, \z: \rho_{2}\der \N: \phi_{2}}$$
and consider
$$\Sigma'_{1} \dem \x_{1}: \sigma_{1}, \y_{1}: \tau_{1}, \z_{1}: \rho_{1}\der \N_{1}: \phi_{1}$$
$$\Sigma'_{2} \dem \x_{2}: \sigma_{2}, \y_{2}: \tau_{2}, \z_{2}: \rho_{2}\der \N_{2}: \phi_{2}$$
where $\N_{1} \equiv \N[\x_{1} / \x, \y_{1} / \y, \z_{1} / \z]$ and $\N_{2} \equiv \N[\x_{2} / \x, \y_{2} / \y, \z_{2} / \z]$,
as copies of $\Sigma_{1}$ and $\Sigma_{2}$ respectively.

Then, from 
${\Gamma,  \x_{1}: \sigma_{1}, \y_{1}: \tau_{1}, \z_{1}: \rho_{1}, \x_{2}: \sigma_{2}, \y_{2}: \tau_{2}, \z_{2}: \rho_{2} \der \M[\N_{1}/\w_{1}, \N_{2}/\w_{2}]: \phi}$
such that $\{ \x, \y, \z \} \cap \dom{\Gamma} = \emptyset$,
we can obtain the following derivation by applying some renaming rules:
\small
\[ \infer[(m)]{\Delta,  \x: \stra{\sigma_{1}, \sigma_{2}}, \y: \stra{\tau_{1}, \tau_{2}}, \z: \stra{\rho_{1}, \rho_{2}} \der \M[\N / \w_{1}, \N / \w_{2}]: \phi}{\infer[(m)]{\Delta,  \x: \stra{\sigma_{1}, \sigma_{2}}, \y: \stra{\tau_{1}, \tau_{2}}, \z_{1}: \rho_{1}, \z_{2}: \rho_{2} \der \M[\N_{1}[\x/\x_{1}, \y/\y_{1}]/\w_{1}, \N_{2}[\x/\x_{2}, \y/\y_{2}]/\w_{2}]: \phi}{
\infer[(m)]{\Delta,  \x: \stra{\sigma_{1}, \sigma_{2}}, \y_{1}: \tau_{1}, \z_{1}: \rho_{1}, \y_{2}: \tau_{2}, \z_{2}: \rho_{2} \der \M[\N_{1}[\x/\x_{1}]/\w_{1}, \N_{2}[\x/\x_{2}]/\w_{2}]: \phi}{\Gamma,  \x_{1}: \sigma_{1}, \y_{1}: \tau_{1}, \z_{1}: \rho_{1}, \x_{2}: \sigma_{2}, \y_{2}: \tau_{2}, \z_{2}: \rho_{2} \der \M[\N_{1}/\w_{1}, \N_{2}/\w_{2}]: \phi}}} \]
\normalsize
\end{example}

In order to state the Substitution Lemma, we introduce the notation
 $S(\Sigma_{1}, ... , \Sigma_{n}, \Pi)$, which stands for the substitution of derivations $\Sigma_{1}, ... , \Sigma_{n}$ in derivation $\Pi$.

\begin{lemma}[Substitution]
\label{lem:subs}
Let $\Pi \dem \Gamma, \x_{1}: \sigma_{1},...,\x_{n}:\sigma_{n}\der \M: \tau$ and $\Sigma_{i} \dem \Delta_{i}\der \N_{i}: \sigma_{i}$, for $1 \leq i \leq n$, such that  ${ \Gamma \# \Delta_{1} \# ... \# \Delta_{n} }$; then there is a clean derivation
${S(\Sigma_{1},...,\Sigma_{n},\Pi) \dem \Gamma,\Delta_{1},...,\Delta_{n}\der \M [ \N_{i}/\x_{i}]_{i=1}^{n}: \tau}$.
\end{lemma}
\begin{proof} 

By induction on the shape of $\Pi$.

Let $\Pi$ be
$$
\infer[(Ax)]{\x:\A \der \x:\A}{}
$$
Then $n=1$, $\Sigma_1 \dem \Delta_1 \der \N_1: \A$, and $S(\Sigma_1, \Pi) = \Sigma_1$.

\medskip

Let $\Pi$ end with an application of rule $(w)$ having range $\{\y\}$. If ${ \y \not\in \{ \x_1,...,\x_n\} }$, then the proof follows by induction. Otherwise, let $\y \equiv \x_1$ and let $\Pi$ be
$$
\infer[(w)]{\Gamma, \x_1: \A ,\x_2: \sigma_2,..., \x_n: \sigma_n\der \M: \tau}{\Pi' \dem \Gamma,\x_2: \sigma_2,..., \x_n: \sigma_n \der \M: \tau & \x_1 \not\in \dom{\Gamma} }
$$
so $\Sigma_1 \dem \Delta_1 \der \N_1: \A$ and $\Sigma_i \dem \Delta_i \der \N_i: \sigma_i$ ($2 \leq i \leq n$).
By induction, there is $S(\Sigma_2,...,\Sigma_n, \Pi) \dem \Gamma, \Delta_2,...,\Delta_{n} \der \subi{\M}{\N_2}{\x_2}{\N_n}{\x_n}: \tau $; then the desired proof is obtained from $S(\Sigma_2,...,\Sigma_n, \Pi)$ by Property~\ref{prop:weak}.

\medskip

Let $\Pi$ end with an application of rule $(\lin I)$: then the result follows easily by induction.

\medskip

Let $\Pi$ be
\[ \infer[(\lin E)]{\Gamma_{1},\Gamma_{2} \der \M\R: \B}
{\Pi_{1} \dem \Gamma_{1}\der \M: \tau \lin \B & 
\Pi_{2} \dem \Gamma_{2}\der \R: \tau & \Gamma_{1} \# \Gamma_{2}} \]
where $\Gamma= \Gamma_1,\Gamma_2$.
Since $\Gamma_1 \# \Gamma_{2}$, we can w.l.o.g. consider a partition of $n$ such that 
${ \Gamma_1= \Gamma'_1, \x_1: \sigma_1,...,\x_k: \sigma_{k} }$ and ${ \Gamma_2= \Gamma'_2, \x_{k+1}: \sigma_{k+1},...,\x_{n}: \sigma_{n} }$;
then the result follows by induction hypothesis and by one application of rule $(\lin E)$.

\medskip

Let $\Pi$ end with an application of rule $(m)$ having range $\{ \y \}$. If ${\y \not\in \{\x_1,...,\x_n\} }$, then the proof follows by induction.
Otherwise, let $\y \equiv \x_1$ and let $\Pi$ be
\[ \infer[(m)]{\Gamma, \x_{1}: \sigma_{1}, \x_{2}: \sigma_{2}, ...\ , \x_{n}: \sigma_{n} \der \M: \tau}
{\Pi' \dem \Gamma, \y_{1}:\rho_{1},...\ ,\y_{h}: \rho_{h},  \x_{2}: \sigma_{2}, ...\ , \x_{n}: \sigma_{n}\der \R:\tau} \]
where $\M \equiv \R[\x_{1}/\y_{k}]_{k=1}^{h}$ and $\sigma_{1} = \sm{\cup}{k=1}{h} \stra{\rho_{k}}$.

Let $\sigma_{1} = \str{\mu}{m}$, for some $m \leq h$.
By Property~\ref{prop:stra}, ${ \Sigma_{1} \dem \Delta_{1} \der \N_{1}: \str{\mu}{m} }$ implies there are $\Psi_{s} \dem \Theta_{s} \der \Q: \mu_{s}$ $(1 \leq s \leq m)$ such that $\Sigma_{1}$ is obtained by an application of rule $(st)$ to $(\Psi_{s})_{s=1}^{m}$, followed by a renaming sequence $\delta$.\\
Observe that, for every $k$ ($1 \leq k\leq h$), there is $s_{k}$ such that $\rho_{k}=\mu_{s_{k}}$; moreover $\rho_{k}=\rho_{k'}$ implies $\Psi_{s_{k}} = \Psi_{s_{k'}}$. 
Let $\Psi'_{s_{k}} \dem \Theta'_{s_{k}} \der \Q_{s_{k}}: \mu_{s_{k}}$ be a copy of $\Psi_{s_{k}}$, for $1 \leq k \leq h$, such that ${(\Gamma, \Delta_{2}, ... , \Delta_{n}) \# \Theta'_{s_{1}} \# ... \# \Theta'_{s_{h}} }$.\\
By induction hypothesis there is
$S(\Psi'_{s_{1}}, ... , \Psi'_{s_{h}}, \Sigma_{2}, ... , \Sigma_{n}, \Pi')$ proving ${\Gamma, \Theta'_{s_{1}}, ... , \Theta'_{s_{h}}, \Delta_{2}, ... , \Delta_{n} \der \R[ \Q_{s_{k}} / \y_{k}]_{k=1}^{h} [\N_{i} / \x_{i}]_{i=2}^{n} : \tau}$, from which, by Lemma~\ref{lem:st-m}, we can derive $\Pi''\dem {\Gamma, \sm{\cup}{k=1}{h} \stra{\Theta_{s_{k}}}, \Delta_{2}, ... , \Delta_{n} \der \R[ \Q / \y_{k}]_{k=1}^{h} [\N_{i} / \x_{i}]_{i=2}^{n} : \tau}$ by a renaming sequence. Note that, since $ \sm{\cup}{k=1}{h} \stra{\Theta_{s_{k}}}= \sm{\cup}{s=1}{m} \stra{\Theta_{s}}$, ${\Pi'' \dem \Gamma, \sm{\cup}{s=1}{m} \stra{\Theta_{s}}, \Delta_{2}, ... , \Delta_{n} \der \R[ \Q / \y_{k}]_{k=1}^{h} [\N_{i} / \x_{i}]_{i=2}^{n} : \tau}$.
Finally, by applying renaming sequence $\delta$ to $\Pi''$, we obtain 
$S(\Sigma_{1}, ... , \Sigma_{n}, \Pi)$.

\medskip

Let $\Pi$ be
\[ \infer[(st)]{\Gamma, \x_1: \sigma_{1} ,...,\x_n: \sigma_{n} \der \M: \tau}
{(\Pi_{k} \dem \Gamma_{k}, \x_{1}: \sigma^{k}_{1}, ... , \x_{n}: \sigma^{k}_{n} \der \M: \tau_k)_{1 \leq k \leq h} } \]

where $\Gamma = \sm{\cup}{k=1}{h} \stra{\Gamma_{k}}$, $\tau = \str{\tau}{h}$ and $\sigma_{i} = \sm{\cup}{k=1}{h} \stra{\sigma^{k}_{i}}$.\\
Let $\sigma_{i} = \stra{\rho^{1}_{i}, ... , \rho^{h_{i}}_{i}}$, where $h_{i} \leq h$.
By Property~\ref{prop:stra}, $\Sigma_{i} \dem \Delta_{i} \der \N_{i}: \stra{\rho^{1}_{i}, ... , \rho^{h_{i}}_{i}}$ implies there are $\Sigma^{s}_{i} \dem \Delta^{s}_{i} \der \R_{i}: \rho^{s}_{i}$ $(1 \leq s \leq h_{i})$ such that $\Sigma_{i}$ is obtained by an application of rule $(st)$ to $(\Sigma^{s}_{i})_{s=1}^{h_{i}}$, followed by a renaming sequence $\delta_{i}$ ${(1 \leq i \leq n)}$.\\
Since by hypothesis ${ \Gamma \# \Delta_{1} \# ... \# \Delta_{n} }$, we can safely assume ${ \Pi \# \Sigma_{1} \# ... \# \Sigma_{n} }$.
For $1 \leq i \leq n$ and $1 \leq k \leq h$, there is $s_{k}$ such that $\sigma^{k}_{i} = \rho^{s_{k}}_{i}$; moreover, $\sigma^{k}_{i} = \sigma^{k'}_{i}$ implies $\Sigma^{s_{k}}_{i} = \Sigma^{s_{k'}}_{i}$.\\
By induction there are ${S(\Sigma^{s_{k}}_{1}, ... , \Sigma^{s_{k}}_{n}, \Pi_{k}) \dem \Gamma_{k},  \Delta^{s_{k}}_{1}, ... , \Delta^{s_{k}}_{n} \der \M[\R_{i}/\x_{i}]_{i=1}^{n}}$ for $1 \leq k \leq h$.
Then $S(\Sigma_{1}, ... , \Sigma_{n}, \Pi)$ is
obtained by applying rule $(st)$ to $( S(\Sigma^{s_{k}}_{1}, ... , \Sigma^{s_{k}}_{n}, \Pi_{k}) )_{1 \leq k\leq h}$,
so obtaining ${\Pi' \dem \Gamma, \sm{\cup}{k=1}{h} \stra{\Theta^{s_{k}}_{1}}, ... , \sm{\cup}{k=1}{h} \stra{\Theta^{s_{k}}_{n}} \der  \M[\R_{i}/\x_{i}]_{i=1}^{n}: \tau}$.
Since $\sm{\cup}{k=1}{h} \stra{\Delta^{s_{k}}_{i}} = \sm{\cup}{s=1}{h_{i}} \stra{\Delta^{s}_{i}}$,
${\Pi' \dem \Gamma, \sm{\cup}{s=1}{h_{1}} \stra{\Theta^{s}_{1}}, ... , \sm{\cup}{s=1}{h_{n}} \stra{\Theta^{s}_{n}} \der  \M[\R_{i}/\x_{i}]_{i=1}^{n}: \tau}$.
Note that, since by hypothesis $\sm{\#}{i=1}{n} \Delta_{i}$, we can apply renaming sequences $\delta_{1}, ... , \delta_{n}$ consecutively to $\Pi'$, so obtaining 
$S(\Sigma_{1}, ... , \Sigma_{n}, \Pi)$.

\medskip

The cases of rules $(\forall I)$ and $(\forall E)$ follow directly by induction.
Note that, since the hypothesis $ \Gamma \# \Delta_{1} \# ... \# \Delta_{n} $, it is immediate to see that $S(\Sigma_{1},...,\Sigma_{n},\Pi)$ is clean.
\end{proof}

Finally we can prove subject reduction. The crucial ingredient, as usual, is the property of detour elimination of a type derivation.
\begin{definition}\label{def:detour}
\begin{enumerate}[(i)]\ 
\item\label{def:foralldet} A $\forall$-detour is a derivation ending with an application of rule $(\forall I)$, immediately followed by an application of rule $(\forall E)$; such detour is eliminated by the following rule:
\small
\[
\begin{array}{lcl}
\infer[(\forall E)]{\Gamma \der \M: \B [\A/\tvar]}
{\infer[(\forall I)]{\Gamma \der \M:\forall \tvar. \B}
{\Pi \dem  \Gamma \der \M: \B \quad \tvar \not\in\dom{\Gamma}}
}
&
\mapsto & \Pi[\A/\tvar] \dem \Gamma \der \M:\B[\A/\tvar]
\end{array}
\]
\normalsize
where $\Pi[\A/\tvar]$ denotes the derivation obtained from $\Pi$ by replacing every occurrence of $\tvar$ by $\A$.
\item\label{def:lindet} A $\lin$-detour is a derivation ending with an application of rule $(\lin I)$, immediately followed by an application of rule $(\lin E)$; such detour is eliminated by the following rule:
\small
\[
\begin{array}{lcl}
\infer[(\lin E)]{\Gamma \der (\lambda \x.\M)\N:\A}
{\infer[(\lin I)]{ \Gamma \der \lambda \x.\M:\sigma\lin \A}
{\Pi \dem\Gamma, \x:\sigma \der \M:\A }
& \Sigma \dem \Delta\der \N:\sigma
}
&
\mapsto & S(\Sigma, \Pi)\dem \Gamma,\Delta \der \M[\N/\x]:\A
\end{array}
\]
\normalsize
where $S(\Sigma, \Pi)$ has been defined in Lemma~\ref{lem:subs}.

\end{enumerate}
\end{definition}
Observe that the operation of $\lin$-detour elimination, as defined before, is not correct: indeed, when applied to a subderivation, it can transform a correct derivation into an incorrect one.
For example, consider a derivation ending with an application of rule $(st)$ to $n \geq 1$ subderivations, whose subject contains a $\beta$-redex: in this case, a $\lin$-detour having the same subject appears in $n$ different subderivations, but eliminating only one of such $\lin$-detours would result in an incorrect derivation.\\
Indeed, one $\beta$-reduction can correspond to many detour eliminations; in practice, this happens both when there are applications of quantifier rules in between the introduction and the elimination of the $\lin$, and when there is an application of rule $(st)$ to $n \geq 1$ subderivations whose subject contains the current $\beta$-redex. Then the reduction of such redex is done by first erasing in sequence all $\forall$-detours, followed by the elimination of all $\lin$-detours  \textsl{simultaneously} in every premise of the application of rule $(st)$. 

\begin{lemma}\label{lem:riarrange}
Any sequence of applications of renaming and quantifier rules can be rearranged in such a way that the  applications of the quantifier rules precede the applications of the renaming rules. 
\end{lemma}
\begin{proof}
Observe that quantifier rules deal with the type, while renaming rules deal with the subject and the variables in the context. Let $(R)$ and $(R')$ be respectively a renaming and a quantifier rule: it is sufficient to prove that the sequence of applications of rules $(R)(R')$ can be replaced by the sequence of applications $(R')(R)$.
\\If $(R)=(m)$ the proof is obvious. If $(R)=(w)$, since the application of $(w)$ may introduce new type variables, the constraints on the application of $(\forall I)$ rule are obviously preserved. 
\end{proof}

\begin{theorem}[Subject Reduction]
$\Gamma  \der \R :\sigma$ and $\R \redbeta \Q$ implies $\Gamma \der \Q: \sigma$.
\end{theorem}
\label{th:subj}
\begin{proof} 
If $\R \redbeta \Q$, then there is a term context $\C[.]$ such that ${\R\equiv \C[(\lam \x. \M) \N]}$ and ${\Q \equiv \C[\M [ \N / \x ] ]}$. The proof is by induction on $\C[.]$.\\
Let $\C[.]\equiv [.]$, and $\Gamma\der (\lambda \x. \M) \N:\sigma$. We proceed by induction on $\sigma$.\\
Let $\sigma$ be a linear type $\A$. By Lemma~\ref{lem:gen}, we can assume that $\Pi$ has the following shape:
\[ \infer=[\delta_{2}]{\Theta, \Delta \der (\lambda \x. \M) \N: \A }
{\infer[(\lin E)]{ \Theta', \Delta' \der (\lambda \x. \M') \N': \A'}{\infer=[\delta_{1}]{\Theta' \der \lambda \x. \M' : \sigma' \lin \A'}{\infer[(\lin I)]{\Theta'' \der \lambda \x. \M'' : \sigma'' \lin \A''}{\Pi' \dem \Theta'', \x: \sigma'' \der \M'' : \A''}} & \Sigma \dem \Delta' \der \N' : \sigma' & \Theta' \# \Delta'}} \]
where $\M$ is an instance of $\M'$ (which, in turn, is an instance of $\M''$) and $\N$ is an instance of $\N'$, $\Gamma = \Theta, \Delta$ and $\delta_{1}, \delta_{2}$ are renaming and quantifier sequences. \\
By Lemma~\ref{lem:riarrange}, $\delta_1$ can be replaced by $\delta',\delta''$, where $\delta'$ contains only quantifier rules and $\delta''$ contains only renaming rules. Moreover by the assumption that all derivations are clean, $\Theta'' \# \Delta'$.
Then we can rewrite $\Pi$ in the following way:
\[ \infer=[\delta_{2}]{\Theta, \Delta \der (\lambda \x. \M) \N: \A }
{\infer=[\delta'']{\Theta', \Delta' \der (\lambda \x. \M') \N': \A'}{\infer[(\lin E)]{ \Theta'', \Delta' \der (\lambda \x. \M'') \N': \A'}{\infer=[\delta']{\Theta'' \der \lambda \x. \M'' : \sigma' \lin \A'}{\infer[(\lin I)]{\Theta'' \der \lambda \x. \M'' : \sigma'' \lin \A''}{\Pi' \dem \Theta'', \x: \sigma'' \der \M'' : \A''}} & \Sigma \dem \Delta' \der \N' : \sigma' & \Theta'' \# \Delta'}}} \]

Let us assume the non-trivial case in which the sequence $\delta'$ is not empty.
Then, since $\sigma' \lin \A'$ is a $\lin$ type, sequence $\delta'$ must end with one application of $(\forall E)$ rule; however, since $\sigma'' \lin \A''$ is also a $\lin$ type, sequence $\delta'$ must contain a matching application of $(\forall I)$ rule: therefore, $\delta'$ contains a $\forall$-detour, which can be eliminated as shown in Definition~\ref{def:detour}.\ref{def:foralldet}; then sequence $\delta_{1}$ decreases by two applications of quantifier rules. By erasing all $\forall$-detours in sequence $\delta'$,
we obtain the following derivation:

\[ \infer=[\delta_{2}]{\Theta, \Delta \der (\lambda \x. \M) \N: \A }
{\infer=[\delta'']{\Theta', \Delta' \der (\lambda \x. \M') \N': \A'}{\infer[(\lin E)]{ \Theta'', \Delta' \der (\lambda \x. \M'') \N': \A'}{\infer[(\lin I)]{\Theta'' \der \lambda \x. \M'' : \sigma' \lin \A'}{\Pi'' \dem \Theta'', \x: \sigma' \der \M'' : \A'} & \Sigma \dem \Delta' \der \N' : \sigma' & \Theta' \# \Delta'}}} \]

Finally, by applying Lemma~\ref{lem:subs} and substituting $\Sigma$ in $\Pi''$ as in Definition~\ref{def:detour}.\ref{def:lindet}, the resulting derivation is

\[ \infer=[\delta_{2}]{\Theta, \Delta \der \M[\N / \x]: \A }{\infer[\delta'']{\Theta', \Delta' \der \M'[\N' / \x]: \A'}{S(\Sigma, \Pi'') \dem \Theta'', \Delta' \der \M'' [ \N' / \x]: \A'}} \]

The case $\delta'$ empty is easier.
Notice that, since the property to be clean is preserved by substitution, the resulting proof is clean.
\medskip

Now let $\sigma$ be a stratified type $\{\sigma_1,...,\sigma_n\}$. By Property~\ref{prop:stra} there are ${\Pi_i \dem \Gamma_i \der (\lambda \x. \M') \N': \sigma_i}$ ($1\leq i \leq n$) and a renaming sequence $\delta$ such that $\Pi$ has the following shape:
\[ \infer=[\delta]{\Gamma \der (\lambda \x. \M) \N:\sigma}{\infer[(st)]{\sm{\cup}{i=1}{n} \stra{\Gamma_{i}} \der (\lambda \x. \M') \N':\sigma }{( \Pi_i\dem\Gamma_i \der (\lambda \x. \M') \N': \sigma_i)_{1\leq i \leq n} }} \]

By inductive hypothesis there are
$\Phi_{i} \dem \Gamma_i \der \M'[\N'/\x]: \sigma_i$, for $1 \leq i \leq n$; then the result follows by applying rule $(st)$ to $\Phi_{1}, ... , \Phi_{n}$, followed by sequence $\delta$.

\medskip

The induction case for $\C[.]\not\equiv[.]$ is straightforward.
\end{proof}
\section{Strong normalization}
\label{sec:sn}
The type assignment system \STR\ characterizes strong normalization, i.e., the following theorem holds.
\begin{theorem}
$\Gamma \der \M: \sigma$, for some $\Gamma$ and $\sigma$, if and only if $\M$ is strongly normalizing.
\end{theorem}
This theorem is the consequence of Theorem~\ref{cor:str1} and Theorem~\ref{cor:str2}, whose proofs are supplied respectively in the next subsections.

\subsection{Typability vs. strong normalization} \label{subsec:tsn}
Given a typable term $\M$, the stratified structure of types and derivations allows to give a bound on the number of $\beta$-reduction steps necessary to reach the normal form of $\M$, which depends both on the size of $\M$ and on the degree of the derivation $\Pi$ (i.e., the nesting of applications of rule $(st)$ in $\Pi$); namely, 
the number of reduction steps is bounded by a polynomial in the size of the term, whose degree depends on the degree of the underlying type derivation. As a consequence, all typable terms in \STR\ are strongly normalizing.

We begin with a few necessary definitions of measure.
\begin{definition}
\label{def:polydef}
\begin{enumerate}\

\item The \textbf{size} $|\M|$ of a term $\M$ is defined inductively as follows:
\[  |\x| = 1 \qquad |\lambda \x. \M| = |\M| + 1 \qquad |\M\N| = |\M| + |\N| + 1 \]

\item The \textbf{rank} of an application of rule $(m)$ with domain $\mathcal{X}$
is the cardinality of the set $\mathcal{X} \cap \FV{\M}$, i.e. the number of variables in the domain of the rule appearing free in $\M$.
Let $r$ be the maximum rank of applications of rule $(m)$ in $\Pi$; then the rank of $\Pi$, denoted by $\rk{\Pi}$, is equal to  $\max \{ 1 , r \}$.

\item The \textbf{degree} of a proof $\Pi$, denoted by $\D{\Pi}$, is the maximal nesting of applications of rule $(st)$ in $\Pi$, i.e. the maximal number of applications of rule $(st)$ in any path connecting the conclusion with an axiom of $\Pi$.

\item Let $r$ be a positive number; the \textbf{weight} $\W{\Pi}{r}$ of $\Pi$ with respect to $r$  is defined inductively as follows:
\begin{itemize}
\item if $\Pi$ ends with an application of rule $(Ax)$, then $\W{\Pi}{r} = 1$;
\item if $\Pi$ ends with an application of rule $(\lin I)$ with premise $\Pi'$, then ${\W{\Pi}{r} = \W{\Pi'}{r} + 1}$;
\item if $\Pi$ ends with an application of rule $(\lin E)$ with premises $\Pi_1$ and $\Pi_2$, then $\W{\Pi}{r} = \W{\Pi_1}{r} + \W{\Pi_2}{r} + 1$;
\item if $\Pi$ ends with an application of rule $(st)$ with premises $(\Pi_{i})_{i=1}^{n}$, then ${\W{\Pi}{r} = r \cdot \max_{i=1}^n \W{\Pi_i}{r}}$;
\item if $\Pi$ ends with an application of a renaming or quantifier rule with premise $\Pi'$, then ${\W{\Pi}{r} = \W{\Pi'}{r}}$.
\end{itemize}

\end{enumerate}
\end{definition}

These measures are related to each other as shown by the following lemma:

\begin{lemma}
\label{lem:polylem}
Let $\Pi \dem \Gamma \der \M: \sigma$. Then:
\begin{enumerate}[(i)]
\item $\rk{\Pi} \leq |\M| $
\item $\W{\Pi}{r} \leq r^{\D{\Pi}} \cdot \W{\Pi}{1}$
\item $\W{\Pi}{1} = |\M|$
\end{enumerate}
\end{lemma}

\begin{proof} All three points are easily proven by induction on $\Pi$.
\end{proof}

\begin{remark}\label{rem:sizeofcopy}
Note that $\Pi'$ being a copy of $\Pi$ implies $\W{\Pi'}{r} = \W{\Pi}{r}$, for every $r \geq 1$ since they have the same structure.
\end{remark}

Now we can state the weighted version of Lemma~\ref{lem:subs}:

\begin{lemma}[Weighted substitution]
\label{lem:wsubs}
Let ${\Pi \dem \Gamma, \x_1: \sigma_1,..., \x_n: \sigma_n\der \M: \tau}$ and ${\Sigma_i \dem \Delta_i \der \N_i: \sigma_i}$, for $1\leq i \leq n$, such that ${\Gamma \# \Delta_{1} \# ... \# \Delta_{n}}$
and ${\{ \x_{1}, ... , \x_{n} \} \cap \FV{\M} = \{ \x_{i_{1}}, ... , \x_{i_{p}} \}}$;
then $${\W{S(\Sigma_{1}, ... , \Sigma_{n}, \Pi)}{r} \leq \W{\Pi}{r} + \sm{\sum}{s=1}{p} \W{\Sigma_{i_{s}}}{r}}$$ for any ${r \geq \max \{ \rk{\Sigma_{1}}, ... , \rk{\Sigma_{n}}, \rk{\Pi} \}}$.
\end{lemma}

\begin{proof} By induction on $\Pi$.  In order to save space and keep the proof simple, we use the same notation as in 
Lemma~\ref{lem:subs}. 

In case $\Pi$ end with an application of rule $(Ax)$, $\W{S(\Sigma_{1}, \Pi)}{r} = \W{\Sigma_{1}}{r}$ and the proof is obvious.

\medskip

Let $\Pi$ end with an application of rule $(w)$, having range $\{ \y \}$, to $\Pi'$. If $\y \not\in \{ \x_{1}, ... , \x_{n} \}$, then the proof follows by induction. Otherwise, let $\y \equiv \x_{1}$, so $\x_{1} \not\in \{ \x_{i_{1}}, ... , \x_{i_{p}} \}$. By inductive hypothesis ${\W{S(\Sigma_{2}, ... , \Sigma_{n}, \Pi')}{r} \leq \W{\Pi'}{r} + \sm{\sum}{s=1}{p} \W{\Sigma_{i_{s}}}{r}}$; then
\small
\begin{align*}
\W{S(\Sigma_{1}, ... , \Sigma_{n}, \Pi)}{r} = \W{S(\Sigma_{2}, ... , \Sigma_{n}, \Pi')}{r} & \leq \W{\Pi'}{r} + \sm{\sum}{s=1}{p} \W{\Sigma_{i_{s}}}{r}\\
\ & = \W{\Pi}{r} + \sm{\sum}{s=1}{p} \W{\Sigma_{i_{s}}}{r} \\
\end{align*}
\normalsize

\medskip

Let $\Pi$ end with an application of rule $(\lin I)$:
then the result follows easily by induction.

\medskip

Let $\Pi$ end with an application of rule $(\lin E)$ to $\Pi_{1}$ and $\Pi_{2}$; moreover,
let ${ \{\x_{1}, ... , \x_{k} \} \cap \FV{\M} = \{ \x_{i_{1}}, ... , \x_{i_{q}} \} }$ and ${ \{\x_{k+1}, ... , \x_{n} \} \cap \FV{\R} = \{ \x_{i_{q+1}}, ... , \x_{i_{p}} \} }$. By inductive hypothesis $\W{S(\Sigma_{1}, ... , \Sigma_{k}, \Pi_{1})}{r} \leq \W{\Pi_{1}}{r} + \sm{\sum}{s=1}{q} \W{\Sigma_{i_{s}}}{r}$ and $\W{S(\Sigma_{k+1}, ... , \Sigma_{n}, \Pi_{2})}{r} \leq \W{\Pi_{2}}{r} + \sm{\sum}{s=q+1}{p} \W{\Sigma_{i_{s}}}{r}$; then
\small
\begin{align*}
\W{S(\Sigma_{1}, ... , \Sigma_{n}, \Pi)}{r} & = \W{S(\Sigma_{1}, ... , \Sigma_{k}, \Pi_{1})}{r} + \W{S(\Sigma_{k+1}, ... , \Sigma_{n}, \Pi_{2})}{r} + 1 \\
\ & \leq \W{\Pi_{1}}{r} + \W{\Pi_{2}}{r} + 1 + \sm{\sum}{s=1}{q} \W{\Sigma_{i_{s}}}{r} + \sm{\sum}{s=q+1}{p} \W{\Sigma_{i_{s}}}{r}\\
\ & = \W{\Pi}{r} + \sm{\sum}{s=1}{p} \W{\Sigma_{i_{s}}}{r}
\end{align*}
\normalsize

\medskip

Let $\Pi$ end with an application of rule $(m)$, having range $\{ \y \}$, to $\Pi'$. If $\y \not\in \{ \x_{i_{1}}, ... , \x_{i_{p}} \}$, then the proof follows by induction. Otherwise let $\y \equiv \x_{1}$, ${ \{ \y_{1}, ... , \y_{h} \} \cap \FV{\R} = \{ \y_{j_{1}}, ... , \y_{j_{q}}\} }$ and ${ \{ \x_{2}, ... , \x_{n} \} \cap \FV{\M} = \{ \x_{i_{1}}, ... , \x_{i_{p}}\} }$.
By inductive hypothesis ${\W{S(\Psi'_{s_{1}}, ... , \Psi'_{s_{h}}, \Sigma_{2}, ... , \Sigma_{n}, \Pi')}{r} \leq \W{\Pi'}{r} + \sm{\sum}{l=1}{q} \W{\Psi'_{s_{j_{l}}}}{r} + \sm{\sum}{s=1}{p} \W{\Sigma_{i_{s}}}{r} }$.\\
Note that ${ \W{\Sigma_{1}}{r} = r \cdot \sm{\max}{s=1}{m}\ \W{\Psi_{s}}{r} = r \cdot \sm{\max}{k=1}{h}\ \W{\Psi'_{s_{k}}}{r} }$ by Remark~\ref{rem:sizeofcopy}.
Moreover, $\W{\Pi}{r} = \W{\Pi'}{r}$ and $q \leq \rk{\Pi} \leq r$.\\
Since $\delta$ is a renaming sequence, by definition we have that ${ \W{S(\Sigma_{1}, ... , \Sigma_{n}, \Pi)}{r} = \W{S(\Psi'_{s_{1}}, ... , \Psi'_{s_{h}}, \Sigma_{2}, ... \Sigma_{n}, \Pi')}{r}}$; then
\small
\begin{align*}
\W{S(\Sigma_{1}, ... , \Sigma_{n}, \Pi)}{r} & = \W{S(\Psi'_{s_{1}}, ... , \Psi'_{s_{h}}, \Sigma_{2}, ... \Sigma_{n}, \Pi')}{r}\\
\ & \leq \W{\Pi'}{r} + \sm{\sum}{l=1}{q} \W{\Psi'_{s_{j_{l}}}}{r} + \sm{\sum}{s=1}{p} \W{\Sigma_{i_{s}}}{r} \\
\ & \leq \W{\Pi'}{r} + r \cdot \sm{\max}{l=1}{q}\ \W{\Psi'_{s_{j_{l}}}}{r} + \sm{\sum}{s=1}{p} \W{\Sigma_{i_{s}}}{r} \\
\ & = \W{\Pi}{r} + r \cdot \sm{\max}{k=1}{h}\ \W{\Psi'_{s_{k}}}{r} + \sm{\sum}{s=1}{p} \W{\Sigma_{i_{s}}}{r} \\
\ & = \W{\Pi}{r} + \W{\Sigma_{1}}{r} + \sm{\sum}{s=1}{p} \W{\Sigma_{i_{s}}}{r}
\end{align*}
\normalsize

\medskip

Let $\Pi$ end with an application of rule $(st)$ to $(\Pi_{k})_{1 \leq k \leq h}$.
Note that ${\W{\Sigma_{i}}{r} = r \cdot \sm{\max}{s=1}{h_{i}}\ \W{\Sigma^{s}_{i}}{r} = r \cdot \sm{\max}{k=1}{h}\ \W{\Sigma^{s_{k}}_{i}}{r} }$ $(1 \leq i \leq n)$.
By induction hypothesis ${ \W{S(\Sigma^{s_{k}}_{1},...,\Sigma^{s_{k}}_{n},\Pi_{k})}{r} \leq \W{\Pi_{k}}{r}+
\sm{\sum}{j=1}{p} \W{\Sigma^{s_{k}}_{i_{j}}}{r} }$, for every ${ 1\leq k\leq h }$; then
\small
\begin{align*}
\W{S(\Sigma^{s_{k}}_{1},...,\Sigma^{s_{k}}_{n},\Pi_{k})}{r} & = r \cdot \sm{\max}{k=1}{h}\ \W{S(\Sigma^{s_{k}}_{1},...,\Sigma^{s_{k}}_{n}, \Pi_{k})}{r} \\
\ & \leq r \cdot \sm{\max}{k=1}{h}\ \left( \W{\Pi_{k}}{r}+
\sm{\sum}{j=1}{p} \W{\Sigma^{s_{k}}_{i_{j}}}{r} \right) \\
\ & < r \cdot \sm{\max}{k=1}{h}\ \W{\Pi_{k}}{r}+
r \cdot \sm{\max}{k=1}{h}\ \sm{\sum}{j=1}{p} \W{\Sigma^{s_{k}}_{i_{j}}}{r} \\
\ & \leq r \cdot \sm{\max}{k=1}{h}\ \W{\Pi_{k}}{r}+ \sm{\sum}{j=1}{p} \left(
r \cdot \sm{\max}{k=1}{h}\ \W{\Sigma^{s_{k}}_{i_{j}}}{r} \right) \\
\ & = \W{\Pi}{r} + \sm{\sum}{j=1}{p} \W{\Sigma_{i_{j}}}{r}\\
\end{align*}
\normalsize

\medskip

Let $\Pi$ end with an application of a quantifier rule: then the result follows by induction.

\end{proof}

Observe that erasing a $\forall$-detour does not change the weight of the proof.
We prove that the weight of a proof strictly decreases with each normalization step:

\begin{lemma}[Weighted subject reduction] \label{lem:wsubj}
$\Pi \dem \Gamma  \der \R :\sigma$ and $\R \redbeta \Q$ implies $\Psi \dem \Gamma \der \Q: \sigma$,
such that
$\W{\Psi}{r} < \W{\Pi}{r}$
for every $r \geq \rk{\Pi}$.
\end{lemma}

\begin{proof}
We already proved in Theorem~\ref{th:subj} that subject reduction holds, so now we just need to prove that the inequality holds. Let $\R\equiv\C[(\lambda\x.\M)\N]$ and ${\Q\equiv\C[\M[\N/\x]]}$. 
We proceed by induction on the term context and then by induction on $\sigma$. In order to save space, we use the same notation as in  Theorem~\ref{th:subj}.
Let ${\C \equiv [.]}$; if $\sigma$ is linear, then ${\W{\Pi}{r} = \W{\Pi'}{r} + \W{\Sigma}{r} + 2 = \W{\Pi''}{r} + \W{\Sigma}{r} + 2}$, since $\delta_{1}, \delta_{2}$ are sequences of renaming and quantifier rules: then
 ${\W{\Psi}{r} = \W{S(\Sigma, \Pi'')}{r} \leq \W{\Sigma}{r} + \W{\Pi''}{r} < \W{\Pi}{r}}$ by Lemma~\ref{lem:wsubs}.

Otherwise, let $\sigma = \{ \sigma_{1}, ... , \sigma_{n} \}$.
By inductive hypothesis $\W{\Phi_{i}}{r} < \W{\Pi_{i}}{r}$, for $1 \leq i \leq n$: then 
$ { \W{\Psi}{r} = r \cdot \sm{\max}{i=1}{n}\ \W{\Phi_{i}}{r} < r \cdot \sm{\max}{i=1}{n}\ \W{\Pi_{i}}{r} = \W{\Pi}{r} }$.

If $\C \not\equiv [.]$, the proof follows easily by induction.
\end{proof}

We can now state both results of complexity and strong normalization.

\begin{lemma} \label{th:polystep}
Let $\Pi \dem \Gamma \der \M: \sigma$ and let $\M$ $\beta$-reduce to $\M'$ in $m$ steps. Then:
\begin{enumerate}[(i)]
\item $m \leq |\M|^{\D{\Pi} + 1}$.
\item $|\M'| \leq |\M|^{\D{\Pi} + 1}$.
\end{enumerate}
\end{lemma}
\begin{proof}
Let $r = \rk{\Pi}$ and let $\Pi' \dem \Gamma \der \M': \sigma$.
\begin{enumerate}[i)]
\item By Lemma~\ref{lem:polylem}, $r \leq |\M|$ and ${\W{\Pi}{r} \leq r^{\D{\Pi}} \cdot \W{\Pi}{1} = r^{\D{\Pi}} \cdot |\M| \leq |\M|^{\D{\Pi}} \cdot |\M|}$.
By Lemma~\ref{lem:wsubs}, if $\Pi$ rewrites to $\Pi_1$ in $1$ step then ${\W{\Pi_1}{r} \leq \W{\Pi}{r} - 1}$; by the same Lemma, if $\Pi_1$ rewrites to $\Pi_2$ in $1$ step then ${\W{\Pi_2}{r} \leq \W{\Pi_1}{r} - 1 \leq \W{\Pi}{r} - 2}$, and so on; then, after $m$ reduction steps we get $\W{\Pi'}{r} \leq \W{\Pi}{r} - m$ and thus $m \leq \W{\Pi}{r}$.
By substituting ${\W{\Pi}{r} \leq |\M|^{\D{\Pi} + 1}}$ in the above expression, we obtain $m \leq |\M|^{\D{\Pi} + 1} $.
\item By Lemma~\ref{lem:wsubs}, ${\W{\Pi'}{r} < \W{\Pi}{r}}$.
Since ${\W{\Pi'}{1} \leq \W{\Pi'}{r}}$ and by Lemma~\ref{lem:polylem} ${|\M'| = \W{\Pi'}{1}}$, we get 
\[ {|\M'| \leq \W{\Pi}{r} \leq r^{\D{\Pi}} \cdot \W{\Pi}{1} = r^{\D{\Pi}} \cdot |\M| \leq |\M|^{\D{\Pi} + 1}} .\]
\end{enumerate}
\end{proof}

Since the number of steps in the reduction path for a given term is a finite number, we also get a proof of strong normalization for all terms typable in \STR.

\begin{theorem}\label{cor:str1}
If a term $\M$ is typed in \STR, then $\M$ is strongly normalizing.
\end{theorem}

Note that any typable term can be assigned an infinite number of types, so every derivation for it supplies a bound on the number of its normalization steps; it is easy to see that every typable term has a minimal typing, which gives the minimal bound on its normalization time.
\subsection{Strong normalization vs. typability} \label{subsec:snt}
In this subsection we show that all strongly normalizing terms are typed in \STR.
We follow the technique used in \cite{RaamsdonkSSX99, Barendregt2013}.\\
The set of strongly normalizing terms, denoted by $\SN$, is the smallest set of terms closed under the following three rules:
\begin{gather}
\infer{\x \M_{1} ... \M_{n} \in \SN}{\M_{1} \in \SN &  ... & \M_{n} \in \SN} \\
\infer{\lam \x. \M \in \SN}{\M \in \SN} \\
\infer{(\lam \x. \M) \N \M_{1} ... \M_{n} \in \SN}{\M[ \N/\x ] \M_{1} ... \M_{n} \in \SN & \N \in \SN}
\end{gather}

For each of these rules, we prove that if the premises of the rule are typable then the conclusion is typable.
We begin by giving a sort of inversion of Lemma~\ref{lem:subs}.

\begin{lemma}\label{lem:invsub}
Let $\Gamma, \Delta \der \M[\N_{i}/\x_{i}]_{i=1}^{n}: \tau$, where
${ \dom{\Gamma} = \FV{\M} \setminus \{ \x_{1}, ... , \x_{n} \} }$,
${ \{ \x_{1}, ... , \x_{n} \} \cap \dom{\Delta} = \emptyset }$ and
${\N_{1}, ... , \N_{n}}$ are typable in $\STR$;
then there are ${\Delta_i \der \N_i:\sigma_i}$, for ${1 \leq i \leq n}$,
such that
${\Gamma, \x_{1}: \sigma_{1},..,\x_{n}: \sigma_{n} \der \M: \tau}$.
\end{lemma}

\begin{proof}
Let $\Pi$ be a derivation proving $\Gamma, \Delta \der \M[\N_{i}/\x_{i}]_{i=1}^{n}: \tau$.
The proof is by induction on $\Pi$.

Let $\Pi$ be
\[ \infer[(Ax)]{\x: \A \der \x: \A}{} \]
so $\M$ is a variable.
By hypothesis $\N_{i}$ is typable, so by Lemma~\ref{lem:gen} there is $\Sigma_{i} \dem \Delta_{i} \der \N_{i}: \A_{i}$ for some linear type $\A_{i}$, for $1 \leq i \leq n$.\\
Let $\M \equiv \x_{k}$, for $1 \leq k \leq n$, so $\Gamma = \emptyset$ and ${\N_{k} \equiv \x}$. Then ${\x_{1}:\A_{1},...,\x_{n}:\A_{n} \der \x_{k}: \A_{k}}$ follows from axiom ${ \x_{k}: \A_{k} \der \x_{k}: \A_{k} }$ by Property~\ref{prop:weak}.\\
Otherwise, let $\M \equiv \x$ and $\x \not\in \{\x_{1}, ... , \x_{n} \}$, so $\Gamma = \x: \A$.
Then ${\x: \A, \x_{1}:\A_{1},...,\x_{n}:\A_{n} \der \x: \A}$ follows from axiom $\x: \A \der \x: \A$ by Property~\ref{prop:weak}.

\medskip

Let $\Pi$ be
\[ \infer[(\lin I)]{\Gamma, \Delta \der  \lambda \y. \R: \rho \lin \A}
 {\Gamma, \Delta, \y:\rho \der \R: \A} \]
where $\lambda \x.\R\equiv\M[\N_{i}/\x_{i}]_{i=1}^{n}$, so either $\M\equiv\lam \y. \R'$ or $\M \equiv \x_{k}$, for some $k\in \{1,...,n\}$.
If ${\M\equiv\lam \y. \R'}$, since ${(\lam \y. \R')[\N_{i}/\x_{i}]_{i=1}^{n} \equiv \lam \y. (\R'[\N_{i}/\x_{i}]_{i=1}^{n})}$, by induction there are 
$\Delta_i \der \N_i: \sigma_i$ such that ${\Gamma, \y: \rho, \x_{1}: \sigma_{1}, ... , \x_{n}: \sigma_{n} \der \R: \A}$ 
($1\leq i\leq n$), so the proof follows by applying rule $(\lin I)$.\\
If $\M \equiv \x_{k}$, then $\N_k\equiv\lambda \y.\R$, ${\A_k = \rho \lin \A}$ and ${\Gamma=\emptyset}$. Since ${\x_k:\rho \lin \A \der \x_k:\rho \lin \A}$ by the axiom rule, ${\x_{1}: \sigma_{1},..,\x_k: \rho \lin \A,...\x_{n}: \sigma_{n}\der \x_k:\rho \lin \A}$ follows by Property~\ref{prop:weak}.

\medskip

Let $\Pi$ be
\[ \infer[(\lin E)]{\Gamma, \Delta \der \R\Q: \A}
{\Gamma', \Delta' \der \R: \rho \lin \A & \Gamma'', \Delta'' \der \Q: \rho} \]
where $\Gamma = \Gamma', \Gamma''$,
so either
$\M\equiv\R' \Q'$ or $\M \equiv \x_{k}$, for some $k\in \{1,...,n\}$.\\
%so $\M$ is either an application or a variable in $\{\x_{1}, ... , \x_{n} \}$.\\
Let $\M \equiv \R' \Q'$, so $\R \Q \equiv \R' [\N_{i}/\x_{i}]_{i=1}^{n}  \Q'[\N_{i}/\x_{i}]_{i=1}^{n}$.
Since ${\R \equiv \R'[\N_{i}/\x_{i}]_{i=1}^{n}}$, by inductive hypothesis there are $\Delta'_i \der \N_i: \sigma'_i$, for ${1 \leq i \leq n}$, such that
${\Gamma', \x_{1}: \sigma'_{1}, ... , \x_{n}: \sigma'_{n} \der \R': \rho \lin \A}$;
moreover, since  ${\Q \equiv \Q'[\N_{i}/\x_{i}]_{i=1}^{n}}$, by inductive hypothesis there are ${\Delta''_i \der \N_i: \sigma''_i}$, for ${1 \leq i \leq n}$, such that
${\Gamma'', \x_{1}: \sigma''_{1}, ... , \x_{n}: \sigma''_{n} \der \Q': \rho}$.
From such derivations, we obtain respectively
${\Sigma' \dem \Gamma', \x'_{1}: \stra{ \sigma'_{1} }, ... , \x'_{n}: \stra{ \sigma'_{n} } \der \R'[\x'_{i}/\x_{i}]_{i=1}^{n}: \rho \lin \A}$ and ${\Sigma'' \dem \Gamma'', \x''_{1}: \stra{ \sigma''_{1} }, ... , \x''_{n}: \stra{ \sigma''_{n} } \der \Q'[\x''_{i}/\x_{i}]_{i=1}^{n}: \rho}$ by a renaming sequence, such that ${ \{ \x'_{1}, ... , \x'_{n} \} \cap \{ \x''_{1}, ... , \x''_{n} \} = \emptyset}$.
For every $1 \leq i \leq n$, let $\sigma_{i} = \stra{ \stra{ \sigma'_{i} } } \cup \stra{ \stra{ \sigma''_{i} } }$: if $\sigma'_{i} \not= \sigma''_{i}$, then we can build
\[ \infer[(st)]{\Delta_{i} \der \N_{i}: \sigma_{i}}{ \infer[(st)]{\stra{ \Delta'_{i} } \der \N_{i}: \stra{ \sigma'_{i} } }{\Delta'_{i} \der \N_{i}: \sigma'_{i} } & \infer[(st)]{\stra{ \Delta''_{i} } \der \N_{i}: \stra{ \sigma''_{i} }}{ \Delta''_{i}  \der \N_{i}: \sigma''_{i} } } \]
where $\Delta_{i} = \stra{ \stra{ \Delta'_{i} } } \cup \stra{ \stra{ \Delta''_{i} } }$; otherwise,
if $\sigma'_{i} = \sigma''_{i}$, then we build the following
\[ \infer[(st)]{\Delta_{i} \der \N_{i}: \sigma_{i}}{ \infer[(st)]{\stra{ \Delta'_{i} } \der \N_{i}: \stra{ \sigma'_{i} } }{\Delta'_{i} \der \N_{i}: \sigma'_{i} } } \]
where $\Delta_{i} = \stra{ \stra{ \Delta'_{i} } }$.
Then there are $\Delta_{i} \der \N_{i}: \sigma_{i}$, for $1 \leq i \leq n$ and a renaming sequence $\delta$ such that the desired derivation is
\small
\[ \infer=[\delta]{\Gamma, \x_{1}: \sigma_{1}, ... , \x_{n}: \sigma_{n} \der \R' \Q': \A}{\infer[(\lin E)]{\Gamma, \x'_{1}: \stra{ \sigma'_{1} }, ... , \x'_{n}: \stra{ \sigma'_{n} }, \x''_{1}: \stra{ \sigma''_{1} }, ... , \x''_{n}: \stra{ \sigma''_{n} } \der \R'[\x'_{i}/\x_{i}]_{i=1}^{n} \Q'[\x''_{i}/\x_{i}]_{i=1}^{n}: \A}{\Sigma' & \qquad \qquad \Sigma''}} \]
\normalsize
If $\M \equiv \x_{k}$, then the proof follows easily as in the previous case $(\lin I)$.
%Now let $\M = \x_{k}$, for $1 \leq k \leq n$, so $\Gamma = \emptyset$ and ${\R \Q = \x_{k}[\N_{i}/\x_{i}]_{i=1}^{n}}$, where $\N_{k} = \R \Q$.
%By hypothesis $\N_{i}$ is typable, so by Lemma~\ref{lem:gen} there is ${\Sigma_{i} \dem \Delta_{i} \der \N_{i}: \A_{i}}$ for some linear type $\A_{i}$, for $1 \leq i \leq n$ and $i \not= k$; moreover, take $\Sigma_{k}$ to be the same derivation as $\Pi$. Then ${\x_{1}: \A_{1},..., \x_{n}:\A_{n} \der \x_{k}: \A}$ follows from the axiom $\x_{k}: \A \der \x_{k}: \A$ by Property~\ref{prop:weak}.

\medskip

Let $\Pi$ be
\[ \infer[(m)]{\Gamma, \Delta, \y: \sm{\cup}{j=1}{m} \stra{\rho_{j}} \der \R[\y/\y_{j}]_{j=1}^{m}: \tau}
{\Gamma, \Delta, \y_{1}: \rho_{1}, ... , \y_{m}: \rho_{m}: \tau_{m} \der \R: \tau } \]
where $ \R[\y/\y_{j}]_{j=1}^{m}\equiv\M[ \N_{i} / \x_{i} ]_{i=1}^{n}$.\\
Let $\y \not \in \sm{\cup}{i=1}{n} \FV{\N_{i}}$, so ${ \M \equiv \M'[ \y / \y_{j}]_{j=1}^{m} }$ and 
${ \R \equiv \M' [ \N_{i} / \x_{i} ]_{i=1}^{n} }$.
If $\y \not\in \FV{\M}$, then ${ \R \equiv \R[\y / \y_{j}]_{j=1}^{m} \equiv \M[\N_{i} / \x_{i}]_{i=1}^{n} }$ and the proof comes by induction. Otherwise, let $\y \in \FV{\M}$ and ${ \{ \y_{s_{1}}, ... , \y_{s_{p}} \} = \{ \y_{1}, ... , \y_{m} \} \cap \FV{\M} }$. By inductive hypothesis there are $\Delta_{i} \der \N_{i}: \sigma_{i}$ such that ${\Gamma, \y_{s_{1}}: \rho_{s_{1}}, ..., \y_{s_{p}}: \rho_{s_{p}}, \x_{1}: \sigma_{1}, ... , \x_{n}: \sigma_{n} \der \M: \tau }$; then, from such derivation, we obtain the desired result by Property~\ref{prop:weak} and by one application of rule $(m)$ with domain $\{ \y_{1}, ... , \y_{n} \}$ and range $\{ \y \}$.\\
Now let us consider $\y \in \sm{\cup}{i=1}{n} \FV{\N_{i}}$, so
${ \M \equiv \M' [ \y / \y_{j} ]_{j=1}^{m} [ \x_{i} / \x^{1}_{i}, ... , \x_{i} / \x^{r_{i}}_{i} ]_{i=1}^{n} } $ and
${\M [\N_{i} / \x_{i} ]_{i=1}^{n} \equiv ( \M' [ \N^{1}_{i} / \x^{1}_{i}, ... , \N^{n_{i}}_{i} / \x^{r_{i}}_{i} ]_{i=1}^{n} ) [\y / \y_{j}]_{j=1}^{m} }$; moreover,
let ${ \{ \y_{s_{1}}, ... , \y_{s_{p}} \} = \{ \y_{1}, ... , \y_{m} \} \cap \FV{\M'} }$.
Since ${ \R \equiv \M' [ \N^{1}_{i} / \x^{1}_{i}, ... , \N^{r_{i}}_{i} / \x^{r_{i}}_{i} ]_{i=1}^{n} }$, by inductive hypothesis there are ${ \Delta^{h}_{i} \der \N^{h}_{i}: \sigma^{h}_{i} }$, for ${ 1 \leq i \leq n }$ and ${ 1 \leq h \leq r_{i} }$, such that \[ { \Psi \dem \Gamma, \y_{s_{1}}: \rho_{s_{1}}, ... , \y_{s_{p}}: \rho_{s_{p}}, \x^{1}_{1}:\sigma^{1}_{1}, ... , \x^{r_{1}}_{1}: \sigma^{r_{1}}_{1}, ... , \x^{1}_{n}: \sigma^{1}_{n}, ... , \x^{r_{n}}_{n}: \sigma^{r_{n}}_{n} \der \M': \tau } \]
Note that ${ \N_{i} \equiv \N^{h}_{i} [ \y / \y_{m} ]_{j=1}^{m} }$, so we can build ${ \Delta'^{h}_{i} \der \N^{h}_{i} [ \y / \y_{m} ]_{j=1}^{m}: \sigma^{h}_{i} }$ by applying rule $(m)$ with domain ${ \FV{\N^{h}_{i}} \cap \{ \y_{1}, ... , \y_{m} \} }$ and range $\{ \y \}$ to ${ \Delta^{h}_{i} \der \N^{h}_{i}: \sigma^{h}_{i} }$, for ${1 \leq i \leq n}$ and ${1 \leq h \leq r_{i}}$.
Then, for every $1 \leq i \leq n$, if ${ \sigma_{i} = \sm{\cup}{h=1}{r_{i}} \stra{ \sigma^{h}_{i} } = \{ \sigma^{s^{i}_{1}}_{i}, ..., \sigma^{s^{i}_{q}}_{i} \} }$ we can build
\[ \infer[(st)]{\Delta_{i} \der \N_{i}: \sigma_{i}}{( \Delta'^{h}_{i} \der \N^{h}_{i} [ \y / \y_{m} ]_{j=1}^{m}: \sigma^{h}_{i} )_{s^{i}_{1} \leq h \leq s^{i}_{q}} } \] 
Let $\delta$ be a renaming sequence containing $n$ applications of rule $(m)$, the $i$-th one having domain $\{ \x^{1}_{i}, ... , \x^{r_{i}}_{i} \}$ and range $\{ \x_{i} \}$, for $1 \leq i \leq n$.
If $\y \in \FV{\M}$, then from $\Psi$ we derive
\[ \Gamma, \y: \sm{\cup}{j=1}{m} \stra{\rho_{j}}, \x_{1}: \sigma_{1}, ... , \x_{n}: \sigma_{n} \der \M' [ \y / \y_{j} ]_{j=1}^{m} [ \x_{i} / \x^{1}_{i}, ... , \x_{i} / \x^{r_{i}}_{i} ]_{i=1}^{n}: \tau \]
by Property~\ref{prop:weak} and by one application of rule $(m)$ with domain $\{ \y_{1}, ... , \y_{n} \}$ and range $\{ \y \}$, followed by sequence $\delta$.
Otherwise, if $\y \not \in \FV{\M}$, then $\M' \equiv \M$ and ${ \{ \y_{s_{1}}, ... , \y_{s_{p}} \} = \emptyset }$, so we obtain the desired result by applying renaming sequence $\delta$ to ${ \Gamma, \x^{1}_{1}:\sigma^{1}_{1}, ... , \x^{r_{1}}_{1}: \sigma^{r_{1}}_{1}, ... , \x^{1}_{n}: \sigma^{1}_{n}, ... , \x^{r_{n}}_{n}: \sigma^{r_{n}}_{n} \der \M: \tau }$.

\medskip

Let $\Pi$ be
\[ \infer[(st)]{\Gamma, \Delta \der \M [ \N_{i} / \x_{i} ]_{i=1}^{n} : \tau}{ ( \Gamma_{j}, \Theta_{j} \der \M [ \N_{i} / \x_{i} ]_{i=1}^{n} : \tau_{j} )_{1 \leq j \leq m} } \]
where $\Gamma = \sm{\cup}{j=1}{m} \stra{\Gamma_{j}}$ and $\tau = \str{\tau}{m}$.
For ${1 \leq j \leq m}$, by inductive hypothesis there are ${\Delta^{j}_{i} \der \N_{i}: \sigma^{j}_{i}}$, for ${1 \leq i \leq n}$, such that $\Gamma_{j}, \x_{1}: \sigma^{j}_{1}, ... , \x_{n}: \sigma^{j}_{n} \der \M: \tau_{j}$.
Then, for $1 \leq i \leq n$, if ${ \sigma_{i} = \sm{\cup}{j=1}{m} \stra{ \sigma^{j}_{i} } = \stra{ \sigma^{s^{i}_{1}}_{i}, ... , \sigma^{s^{i}_{p}}_{i} } }$ there are
\[ \infer[(st)]{\Delta_{i} \der \N_{i}: \sigma_{i} }{ ( \Delta^{j}_{i} \der \N_{i}: \sigma^{j}_{i} )_{s^{i}_{1} \leq j \leq s^{i}_{p}}}\]
such that the desired derivation is
\[\infer[(st)]{\Gamma, \x_{1}: \sigma_{1}, ... , \x_{n}: \sigma_{n} \der \M: \tau }{ ( \Gamma_{j}, \x_{1}: \sigma^{j}_{1}, ... , \x_{n}: \sigma^{j}_{n} \der \M: \tau_{j} )_{1 \leq j \leq m}} \]

\medskip
Let $\Pi$ end with an application of a quantifier rule: then the proof follows easily by induction.
 
\end{proof}

The following lemma proves a particular case of the (typed) subject expansion. 

\begin{lemma}
\label{lem:subexp}
$\Theta \der \M [\N / \x]: \sigma$ and $\N$ typable in \STR\ imply there exists $\Theta'$ such that
$\Theta' \der (\lambda \x.\M)\N: \sigma$.
\end{lemma}

\begin{proof}
By induction on $\sigma$.

Let $\sigma = \A$;
consider $\Theta = \Gamma, \Xi$ where $\dom{\Gamma} = \FV{\M} \setminus \{ \x \}$.
By Lemma~\ref{lem:invsub}, there is $\Delta \der \N: \tau$ such that $\Gamma, \x: \tau \der \M: \A$. Let $\Delta' \der \N': \tau$ be a copy of $\Delta \der \N: \tau$, such that $\Gamma \# \Delta'$: then we can build the following derivation
\[ \infer[(\lin E)]{\Gamma, \Delta' \der (\lam \x. \M) \N': \A}{\infer[(\lin I)]{\Gamma \der \lam \x. \M: \tau \lin \A}{\Gamma, \x: \tau \der \M: \A}  &  \Delta' \der \N': \tau} \]
and, since $(\lam \x. \M) \N$ is an instance of $(\lam \x. \M) \N'$, the result follows by Remark~\ref{rem:instance}.\\
Otherwise, let $\sigma = \str{\sigma}{n}$. By Property~\ref{prop:stra}, $\Pi \dem \Gamma \der \M [ \N / \x ]: \str{\sigma}{n}$ implies there are derivations $\Pi_{i} \dem \Gamma_{i} \der \R: \sigma_{i}$ $(1 \leq i \leq n)$
such that $\Pi$ is obtained by an application of rule $(st)$ to $(\Pi_{i})_{i=1}^{n}$, followed by a renaming sequence $\delta$.
By applying sequence $\delta$ to $\Pi_{i}$ we obtain ${ \Gamma''_{i} \der \M [ \N / \x ] : \sigma_{i} }$, for $1 \leq i \leq n$.
By inductive hypothesis there are $\Gamma'_{i}$ such that $\Gamma'_{i} \der (\lam \x. \M) \N : \sigma_{i}$, for $1 \leq i \leq n$: then by applying rule $(st)$ to such derivations we obtain $\Gamma' \der (\lam \x. \M) \N: \sigma$, where $\Gamma' = \sm{\cup}{i=1}{n} \stra{\Gamma'_{i}}$.
\end{proof}

\setcounter{equation}{0}

Now we can finally prove the desired result.

\begin{theorem} \label{cor:str2} If $\M$ is strongly normalizing, then $\M$ is typable in \STR.
\end{theorem}
\begin{proof}
For each of the three rules defining the set \SN, we show that if the premises of the rule are typable then the conclusion is typable in \STR.

Let us consider rule (1), so ${ \M \equiv \x\N_{1}...\N_{n} }$. Let ${\x',\N'_{1}, ... , \N'_{n}}$ be instances of ${\x,\N_{1},...,\N_{n}}$ respectively, such that ${ \x' \not\in \sm{\cup}{i=1}{n} \FV{ \N'_{i} } }$ and $j \not= h$ implies ${\FV{\N'_{j}} \cap \FV{\N'_{h}} = \emptyset}$, for ${1 \leq j, h \leq n}$: we prove that there is a derivation, assigning to $\x'\N'_{1}...\N'_{n}$ a linear type, from which we can derive a typing for $\x\N_{1}...\N_{n}$.
If $n=0$, then the proof is obvious.
Otherwise, by inductive hypothesis on $\N'_1, ..., \N'_n$ and Corollary~\ref{cor:stra}, there are $\Gamma_i' \der \N'_i : \A_i$, for $1 \leq i \leq n$. Let ${\A = \A_1 \lin \A_2 \lin ... \lin \A_n \lin \B}$; then we can build
\[ \infer[(\lin E)]{\x': \A, \Gamma'_1, ..., \Gamma'_n \der \x' \N'_1 ... \N'_n : \B}{{\infer[(\lin E)]{\vdots}{\infer[(\lin E)]{\x': \A, \Gamma'_1 \der \x' \N'_1 : \A_2 \lin ... \lin \A_n \lin \B}{\infer[(Ax)]{\x': \A \der \x': \A}{} & \Gamma'_1 \der \N'_1 : \A_1} &\Gamma'_2 \der \N'_2 : \A_2}}} \]
and, since $\x \N_{1} ... \N_{n}$ is an instance of $\x' \N'_{1} ... \N'_{n}$, the desired derivation follows by Remark~\ref{rem:instance}.

Now let us consider rule (2), so $\M \equiv \lambda \x. \N$. Since $\N \in \SN$, by inductive hypothesis on $\N$ there is a derivation $\Gamma\der \N : \B$. If $\x \in \FV{\N}$, then $\Gamma = \Gamma', \x: \tau$, so the desired derivation follows by one application of rule $(\lin I)$ to $\Pi'$; otherwise, we first apply a rule $(w)$ with range $\{\x\}$, followed by an application of rule $(\lin I)$ to abstract over $\x$.

Finally, let us consider rule (3), so ${ \M\equiv(\lam \x. \M) \N \M_{1} ... \M_{n} }$.
By induction, both ${ \M [\N/\x] \M_{1} ... \M_{n} }$ and $\N$ are typable.
By Lemma~\ref{lem:gen}, the derivation proving ${\Gamma \der \M[ \N / \x] \M_{1} ... \M_{n}: \A}$ is
\[ \infer=[\delta_{n}]{\Gamma \der \M[ \N / \x] \M_{1} ... \M_{n}: \A }{\deduce[\vdots]{\Theta_{n}, \Delta_{n} \der \R_{n} \M^{n}_{1} ... \M^{n}_{n}: \B_{n} }{\infer=[\delta_{2}]{\Theta_{3} \der \R_{3} \M^{3}_{1} \M^{3}_{2}: \sigma^{3}_{3} \lin ... \lin \sigma^{3}_{n} \lin \B_{3}}{\infer[(\lin E)]{\Theta_{2}, \Delta_{2} \der \R_{2} \M^{2}_{1} \M^{2}_{2}: \sigma^{2}_{3} \lin ... \lin \sigma^{2}_{n} \lin \B_{2}}{\infer=[\delta_{1}]{\Theta_{2} \der \R_{2} \M^{2}_{1}: \sigma^{2}_{2} \lin ... \lin \sigma^{2}_{n} \lin \B_{2}}{\infer[(\lin E)]{\Theta_{1}, \Delta_{1} \der \R_{1} \M^{1}_{1}: \sigma^{1}_{2} \lin ... \lin \sigma^{1}_{n} \lin \B_{1}}{\Theta_{1} \der \R_{1}: \sigma^{1}_{1} \lin ... \lin \sigma^{1}_{n} \lin \B_{1}   &   \Delta_{1} \der \M^{1}_{1}: \sigma^{1}_{1} }} & \Delta_{2} \der \M^{2}_{2}: \sigma^{2}_{2}}}}} \]
where $\M [ \N / \x ]$ is an instance of $\R_{i}$, $\M_{i}$ is an instance of $\M^{j}_{i}$ and $\delta_{i}$ is a renaming and quantifier sequence, for $1 \leq i \leq n$ and $i \leq j \leq n$.
We assume w.l.o.g. that ${ \FV{\M} \cap \sm{\cup}{i=1}{n} \FV{\M^{i}_{i}} = \emptyset }$ and ${\FV{\N} \cap \sm{\cup}{i=1}{n} \FV{\M^{i}_{i}} = \emptyset }$.
By Lemma~\ref{lem:riarrange}, each sequence $\delta_{i}$ can be rearranged in such a way that the applications of quantifier rules precede the applications of renaming rules; let $\delta'_{i}$ be such renaming sequence.
Since $\M [ \N/ \x]$ is an instance of $\R_{1}$, by Remark~\ref{rem:instance} we can build ${\Theta \der \M[\N / \x]: \sigma^{1}_{1} \lin ... \lin \sigma^{1}_{n} \lin \B_{1}}$ from ${\Theta_{1} \der \R_{1}: \sigma^{1}_{1} \lin ... \lin \sigma^{1}_{n} \lin \B_{1}}$.
By hypothesis $\N$ is typable in \STR, so by Lemma~\ref{lem:subexp} there is $\Theta'$ such that ${\Theta' \der (\lam \x. \M) \N: \sigma^{1}_{1} \lin ... \lin \sigma^{1}_{n} \lin \B_{1}}$; then we can build
\[ \infer=[\delta'_{n}]{\Theta', \Delta_{1}, ... , \Delta_{n} \der (\lam \x. \M) \N \M^{1}_{1} ... \M^{n}_{n}: \A }{\deduce[\vdots]{\Theta', \Delta_{1}, ... , \Delta_{n} \der (\lam \x. \M) \N \M^{1}_{1} ... \M^{n}_{n}: \B_{n} }{\infer=[\delta'_{2}]{\Theta', \Delta_{1}, \Delta_{2} \der (\lam \x. \M) \N \M^{1}_{1} \M^{2}_{2}: \sigma^{3}_{3} \lin ... \lin \sigma^{3}_{n} \lin \B_{3}}{\infer[(\lin E)]{\Theta', \Delta_{1}, \Delta_{2} \der (\lam \x. \M) \N \M^{1}_{1} \M^{2}_{2}: \sigma^{2}_{3} \lin ... \lin \sigma^{2}_{n} \lin \B_{2}}{\infer=[\delta'_{1}]{\Theta', \Delta_{1} \der (\lam \x. \M) \N \M^{1}_{1}: \sigma^{2}_{2} \lin ... \lin \sigma^{2}_{n} \lin \B_{2}}{\infer[(\lin E)]{\Theta', \Delta_{1} \der (\lam \x. \M) \N \M^{1}_{1}: \sigma^{1}_{2} \lin ... \lin \sigma^{1}_{n} \lin \B_{1}}{\Theta' \der (\lam \x. \M) \N: \sigma^{1}_{1} \lin ... \lin \sigma^{1}_{n} \lin \B_{1}   &   \Delta_{1} \der \M^{1}_{1}: \sigma^{1}_{1} }} & \Delta_{2} \der \M^{2}_{2}: \sigma^{2}_{2}}}}} \]
Since $(\lam \x. \M) \N \M_{1} ... \M_{n}$ is an instance of $(\lam \x. \M) \N \M^{1}_{1} ... \M^{n}_{n}$, the proof desired derivation follows by Remark~\ref{rem:instance}.
\end{proof}
\section{Polynomial characterization} \label{sec:poly}
In this section we prove that \STR\ is sound and complete with respect to \FPTIME; therefore,
while having more typability power,
\STR\ characterizes exactly the same functions as the Soft Type Assignment System \STA\ \cite{GaboardiRonchi07csl}, 
which was proved to be sound and complete with respect to \FPTIME. 

\subsection{From \STA\ to \STR} \label{subsec:trad}

Let us briefly recall the type assignment system \STA.

\begin{definition}
\begin{itemize}
\item The set $\TS$ of \STA-types is defined by the following syntax:
$$\U::= \tvar \mid \mu \lin \U \mid \forall \tvar.\U \quad \mbox{ linear types}$$ 
$$\mu, \nu ::= \U \mid !\mu \quad \mbox{ modal types }$$
\item Derivations in \STA\ assign \STA-types to $\lambda$-terms. The rules are given in Table~\ref{table:sta}.
We extend to \STA\ the notations we already introduced for \STR.
\item Measures in \STA\  are defined in a similar way as in \STR. 
\begin{itemize}
\item The \textbf{rank} of a multiplexor rule $(m)$
$$\infer[(m)]{\Theta, \x:!\mu\ders \sub{\M}{\x}{\x_{1},...,\x_{n}}:\nu}
{\Theta,\x_1:\mu,\ldots,\x_n:\mu\ders \M:\nu}$$
is the cardinality of the set $\{ \x_{1}, ... , \x_{n} \} \cap \FV{\M}$.
Let $r$ be the maximum rank of all rules $(m)$ in $\Pi$; then the rank $\rks{\Pi}$ of $\Pi$ is the maximum between $1$ and $r$.

\item The \textbf{degree} of a proof $\Pi$, denoted by $\DS{\Pi}$, is the maximal nesting of applications of the $(sp)$ rule in $\Pi$, i.e. the maximal number of applications of the $(sp)$ rule in any path connecting the conclusion with some axiom of $\Pi$.

\item Let $r$ be a positive integer. The \textbf{weight} $\WS{\Pi}{r}$ of $\Pi$ with respect to $r$ is defined inductively as follows:
\begin{itemize}
\item if $\Pi$ ends with an application of rule $(Ax)$, then $\WS{\Pi}{r} = 1$;
\item if $\Pi$ ends with an application of rule $(\lin I)$ with premise $\Pi'$, then $\WS{\Pi}{r} = \WS{\Pi'}{r} + 1$;
\item if $\Pi$ ends with an application of rule  $(\lin E)$ with premises $\Pi_1$ and $\Pi_2$, then $\WS{\Pi}{r} = \WS{\Pi_1}{r} + \WS{\Pi_2}{r} + 1$;
\item if $\Pi$ ends with an application of rule $(sp)$ with premise $\Pi'$, then $\WS{\Pi}{r} = r \cdot \WS{\Pi'}{r}$
\item in every other case, $\WS{\Pi}{r} = \WS{\Pi'}{r}$ where $\Pi'$ is the premise of the rule.
\end{itemize}

\end{itemize}
\end{itemize}
\end{definition}

\begin{table}
\begin{center}
\begin{tabular}{|c|}
\hline
\\
\small$
\infer[(Ax)]{\x: \U\ders \x: \U}{}
\quad
\infer[(w)]{\Theta,\x: \U \ders \M:\mu}{\Theta \ders \M:\mu &\x \not \in dom{\Theta} }
\quad
\infer[(\lin I)]{\Theta\ders \lambda \x.\M:\mu\lin \U}{\Theta, \x:\mu\ders \M:\U }
$
\\
\\
$\infer[(\lin E)]{\Theta,\Xi\ders \M\N:\U}{\Theta\ders \M:\mu\lin \U & \Xi \ders \N:\mu & \Theta \# \Xi}
\quad 
\infer[(\forall I)]{\Theta\ders {\tt M}:\forall \tvar .\U}
{\Theta\ders \M:\U & \tvar\not \in FV(\Theta)}
$
\\
\\
$
\infer[(\forall E)]{\Theta\ders\M:\sub{\B}{\U}{\tvar}}
{\Theta\ders \M:\forall \tvar .\B }
\quad
\infer[(m)]{\Theta, \x:!\mu\ders \sub{\M}{\x}{\x_{1},...,\x_{n}}:\nu}
{\Theta,\x_1:\mu,\ldots,\x_n:\mu\ders \M:\nu}
\quad
\infer[(sp)]{!\Theta\ders \M:!\mu}{\Theta\der \M:\mu}
$\normalsize
\\
\\
\hline
\end{tabular}
\end{center}
\caption{The Soft Type Assignment (\STA) system}
\label{table:sta}
\end{table} 

Here we recall the key technical property of \STA, which is very similar to the property for \STR\ proved in Theorem~\ref{th:polystep}.
\begin{property}{\rm \cite{GaboardiRonchi07csl}}
Let $\Pi \dem \Theta \ders \M: \mu$ and $\M\redbetas\M'$ in $m$ steps. Then
\begin{enumerate}
\item $m \leq |\M|^{\DS{\Pi} + 1}$.
\item $|\M'| \leq |\M|^{\DS{\Pi} + 1}$.
\end{enumerate}
\end{property}

A translation from \STA\ to \STR\ can be obtained in a straightforward way: indeed, \STA\ can be seen as a restriction of \STR, where only variables with the same type can be contracted.

\begin{definition}\label{def:SI}\ 
\begin{itemize}
\item The translation $\SI{.}$ from $\TS$ to $\T$ is defined as:
$$\SI{\tvar}=\tvar; \quad \SI{\mu \lin \U}=\SI{\mu}\lin \SI{\U}; \quad \SI{ !\mu}=\stra{\SI{\mu}}$$
\item Let $\Theta$ be a context in \STA. $\SI{\Theta}$ is the context such that $\SI{\Theta}(x)= \SI{\Theta(x)}$.
\end{itemize}
\end{definition}

\begin{lemma}\label{lem:sta-str}
$\Theta \ders \M:\mu$ implies $\SI{\Theta} \der \M: \SI{\mu}$.
\end{lemma}
\begin{proof}
By induction on the derivation. If the last rule is $(Ax)$ the proof is trivial; all other cases follow easily by induction and by applying the respective rule in \STR;
in particular, if $\Pi$ ends with an application of rule $(sp)$ to $\Pi'$, then by induction hypothesis $\SI{\Pi'} \dem \SI{\Theta} \ders \M: \SI{\nu}$, so $\SI{\Pi}$ is obtained by applying rule $(st)$ to $\SI{\Pi'}$.
\end{proof}

\subsection{Polynomial soundness and completeness} \label{subsec:W}

We represent natural numbers in binary notation, following the Church representation of binary words, in which the natural number $0$ is represented by the term $\num{0} \equiv \lambda \s{0} \s{1} \x. \x$ and the natural number $n$ is represented by the term
$\num{n}\equiv\lambda \s{0} \s{1} \x. \s{i_{1}}( ... (\s{i_{m}} \x) ... )$, where the binary representation of $n$ is $\langle i_{1}...i_{m} \rangle$ ($i_{j} \in \{ 0,1 \}$ for $1 \leq j \leq m$).
\begin{example}
The natural number $6$ is represented by $\num{6} \equiv \lam \s{0} \s{1} \x. \s{1} (\s{1} (\s{0} \x))$, while $9$ is represented by $\num{9} \equiv \lam \s{0} \s{1} \x. \s{1} (\s{0} (\s{0} (\s{1} \x)))$.
\end{example}

In \STA, natural numbers in binary notation can be assigned both the uniform type ${ \Wtype= \forall \tvar.!(\tvar \lin \tvar) \lin !(\tvar \lin \tvar) \lin \tvar \lin \tvar }$ and a parametric type
${ \Wtypei{n}{m}=\forall \tvar.!^{n}(\tvar \lin \tvar) \lin !^{m}(\tvar \lin \tvar) \lin \tvar \lin \tvar }$, for any $n \geq 1$ and $m \geq 1$; note that $\Wtype = \Wtypei{1}{1}$. It is easy to check that every derivation $\Pi \dem \emptyset \ders \num{w}: \Wtypei{n}{m}$, for $n \geq 1$ and $m \geq 1$, is such that $\DS{\Pi} = 0$: indeed, the polynomiality of \STA\ depends on this very property.

The parametric numeral types play an essential role, since a term representing a numerical function can have different parameters for its input(s) and output types; in that case, the iteration of such functions is forbidden, with the result that terms representing non-polynomial functions (like exponentiation) cannot be typed. 
Analogously, binary numbers can be assigned in \STR\ both the uniform type ${ \Witype = \forall \tvar.\stra{\tvar \lin \tvar} \lin \stra{\tvar \lin \tvar} \lin \tvar \lin \tvar }$ and a parametric type 
${ \Witypei{n}{m}=\forall \tvar.\{ \tvar \lin \tvar\}^{n} \lin \{ \tvar \lin \tvar\}^{m} \lin \tvar \lin \tvar }$, for any $n \geq 1$ and $m \geq 1$. Again, it is easy to check that any derivation $\Pi \dem \emptyset \ders \num{w}:\Witypei{n}{m}$ is such that $\D{\Pi} = 0$, for $n \geq 1$ and $m \geq 1$.

Now we are able to formally define the representation of functions in both \STA\ and \STR. This definition is a straightforward extension of the classical definition of $\lambda$-representation of functions \cite{Barendregt84, ronchi04book} in a typed setting \cite{BucciaPiper03, GaboardiRonchi07csl}.
The additional power of \STR\ with respect to \STA\ is displayed by the fact that, while in \STA\ every input data must be assigned a parametric type for natural numbers, in \STR\ we allow every input number to have a set of types, with the proviso that all its linear components are numerical types.

\begin{definition}\label{def:num}
Let a {\em program} be a closed term in normal form, and let ${ \phi : \Num^{p}\longrightarrow \Num }$ be a function of arity $p$. 
\begin{enumerate}[(i)]
\item\label{def:numSTA} A program $\M \equiv \lam \x_{1} ... \x_{p}. \R$ represents $\phi$ in \STA\ if and only if:
\begin{itemize}
\item ${\M \num{n_{1}}...\num{n_{p}}=\num{\phi(n_{1},...,n_{p})} }$;
\item ${ \ders \M: !^{i_{1}} \Wtypei{j_{1}}{k_{1}} \lin ... \lin !^{i_p} \Wtypei{j_{p}}{k_{p}} \lin \Wtypei{j}{k} }$, for some ${ j, k, j_{h}, k_{h} }$ ${(1\leq h \leq p)}$.
\end{itemize}
\item\label{def:numSTR} A program $\M \equiv \lam \x_{1} ... \x_{p}. \R$ represents $\phi$ in \STR\ if and only if:
\begin{itemize}
\item ${ \M\num{n_{1}}...\num{n_{p}}=\num{\phi(n_{1},...,n_{p})} }$;
\item ${ \der \M:  \sigma_{1} \lin ... \lin \sigma_{p} \lin \Witypei{h}{k} }$
and
${ \bar{\sigma_{i}} = [ \Witypei{h^{i}_{1}}{k^{i}_{1}},...,\Witypei{h^{i}_{q_{i}}}{k^{i}_{q_{i}}}  ] }$,
for some $h, k, h_{i}, k_{i}, q_{i}, h^{i}_{r}, k^{i}_{r}$ $(1 \leq i \leq p, 1 \leq r \leq q_{i})$.
\end{itemize}
\end{enumerate}
\end{definition}

In \cite{GaboardiRonchi07csl}, the authors proved that \STA\ is sound and complete with respect to \FPTIME; here we exploit the translation from \STA\ to \STR\ in order to extend the completeness result to \STR.

\begin{lemma}[\FPTIME\ completeness]\label{lem:complSTR}
All polynomial functions are representable in \STR.
\end{lemma}
\begin{proof}
By Definition~\ref{def:num}.(\ref{def:numSTA}), let $\M$ be a term representing the polynomial function $\phi: \Num^{p}\longrightarrow \Num$ in \STA, such that
$\M \num{n_{1}}...\num{n_{p}}\equiv\num{\phi(n_{1},...,n_{p})}$ and
${ \ders \M: !^{i_1} \Wtypei{j_{1}}{k_{1}} \lin ... \lin !^{i_p} \Wtypei{j_{p}}{k_{p}} \lin \Wtypei{j}{k} }$,
for some ${ j, k, j_{h}, k_{h} }$ ${(1\leq h \leq p)}$.\\
By Definition~\ref{def:SI}, we know that ${ \SI{!^{i_s} \Wtypei{j_{s}}{k_{s}}}=\{ \Witypei{j_{s}}{k_{s}}\}^{i_s}}$, for ${1 \leq s \leq p}$, and ${ \SI{\Wtypei{j}{k}} = \Witypei{j}{k} }$: then,
by Lemma~\ref{lem:sta-str}, the above \STA\ derivation is translated to
${ \der \M: \{ \Witypei{j_{1}}{k_{1}} \}^{i_1} \lin ... \lin \{ \Witypei{j_{p}}{k_{p}}\}^{i_p} \lin \Witypei{j}{k} }$,
so $\M$ represents $\phi$ in \STR.
\end{proof}

Let $\M$ be a term representing a numerical function $\phi : \Num^{p}\longrightarrow \Num$ in \STR;
in order to prove the soundness of \STR\ w.r.t. \FPTIME, we show that the reduction of $\M \num{n_{1}} ... \num{n_{p}}$ to its normal form can be performed on a Turing Machine of time polynomial in the size of the input. First we need to introduce the notion of ancestors, to keep track of the axioms introducing a given variable in the context.

\begin{definition}[Ancestors]
Let $\Pi\dem\Gamma \der \M: \tau$ and $\x \in \dom{\Gamma}$; the set of ancestors of $\x$ in $\Pi$, denoted by $A(\x,\Pi)$, is defined inductively as follows:
\begin{itemize}
\item if $\Pi$ is $$\infer[(Ax)]{\x: \A \der \x: \A}{}$$ then $A(\x,\Pi )=\{\x\}$;
\item if $\Pi$ is $$\infer[(w)]{\Gamma, \y: \A \der \M: \tau}{\Pi' \dem \Gamma\der \M:\tau}$$ then  $A(\x,\Pi )=\{\x\}$ if $\y\equiv\x$, $A(\x,\Pi )=A(\x,\Pi' )$ otherwise;
\item if $\Pi$ is $$\infer[(\lin I)]{\Gamma \der \lambda \y.\R:\sigma\lin \A}{\Pi'\dem\Gamma, \y: \sigma \der \R:\A }$$ then $A(\x,\Pi )=A(\x,\Pi' )$;
\item if $\Pi$ is $$\infer[(\lin E)]{\Gamma', \Gamma'' \der \R \Q : \A}{\Pi' \dem \Gamma' \der \R: \sigma \lin \A & \Pi'' \dem \Gamma'' \der \Q: \sigma }$$ then $A(\x,\Pi )=A(\x,\Pi' )$ if ${\x \in \dom{\Gamma'}}$, $A(\x,\Pi )=A(\x,\Pi'' )$ otherwise;
\item if $\Pi$ is $$\infer[(m)]{\Gamma, \y: \sm{\cup}{i=1}{n} \stra{\sigma_{i}} \der \M[\y/\y_{i}]_{i=1}^{n}: \tau}{\Pi' \dem \Gamma, \y_1:\sigma_1,...,\y_n:\sigma_n\der \M: \tau}
$$ then $A(\x,\Pi)= \sm{\cup}{i=1}{n} A(\y_i,\Pi')$ if $\x \equiv \y$, $A(\x, \Pi) = A(\x, \Pi')$ otherwise;
\item if $\Pi$ is
$$\infer[(st)]{\sm{\cup}{i=1}{n} \stra{\Gamma_{i}} \der \M: \str{\sigma}{n}}
{(\Pi_i \dem \Gamma_{i} \der \M:\sigma_{i})_{1 \leq i \leq n} }$$
then $A(\x,\Pi)=\sm{\cup}{i=1}{n} A(\x,\Pi_i)$;
\item if $\Pi$ is
$$\infer[(R)]{\Gamma \der \M: \sigma}{\Pi' \dem\Gamma \der \M':\sigma'}$$
where $(R)$ is a quantifier rule, then $A(\x,\Pi)=A(\x,\Pi')$.
\end{itemize} 
\end{definition}

\begin{theorem}[\FPTIME\ soundness]\label{th:polbound}
Let $\phi : \Num^{p}\longrightarrow \Num$ and let $\M$ be a program representing $\phi$ in \STR; then 
$\M \num{n_{1}}...\num{n_{p}}$ can be evaluated to its normal form on a Turing Machine in time $O(P(n_{1} + ... + n_{p}))$, for some polynomial $P$.
\end{theorem}
\begin{proof}
By Definition~\ref{def:num}.(\ref{def:numSTR}), if a program $\M \equiv \lam \x_{1} ... \x_{p}. \R$ represents $\phi$ in \STR, then ${\M \num{n_{1}}...\num{n_{p}}\equiv\num{\phi(n_{1},...,n_{p})} }$ and ${\der \M: \sigma_{1} \lin ... \lin \sigma_{p} \lin \Witypei{h}{k} }$, where
${ \bar \sigma_{i}= [ \Witypei{h^{i}_{1}}{k^{i}_{1}},...,\Witypei{h^{i}_{q_{i}}}{k^{i}_{q_{i}}}  ] }$
for some $h, k, h_{i}, k_{i}, q_{i}, h^{i}_{r}, k^{i}_{r}$ $({1 \leq i \leq p,} {1 \leq r \leq q_{i}})$.
By Lemma~\ref{lem:gen}, we can assume w.l.o.g that ${\M: \sigma_{1} \lin ... \lin \sigma_{p} \lin \Witypei{h}{k} }$ ends with $p$ applications of rule $(\lin I)$ with premise ${\Phi \dem \x_{1}: \sigma_{1}, ... , \x_{p}: \sigma_{p} \der \R: \Witypei{h}{k} }$.
Then we can build the following derivation:
\[ \infer[(\lin E)] {\der \M \num{n_{1}}...\num{n_{p}} : \Witypei{h}{k} }{ \infer[(\lin E)]{\deduce{ \der \M \num{n_{1}}...\num{n_{p-1}}  : \sigma_{p} \lin \Witypei{h}{k}}{\vdots}}{\infer=[(\lin I)]{\der \M: \sigma_{1} \lin ...\lin \sigma_{p} \lin \Witypei{h}{k} }{\Phi \dem \x_{1}: \sigma_{1}, ... , \x_{p}: \sigma_{p} \der \R: \Witypei{h}{k} } & \Phi_{1} \der \num{n_{1}}: \sigma_{1}} & \Phi_{p} \der \num{n_{p}}: \sigma_{p} } \]
where each $\Phi_{i}$ ${(1 \leq i \leq p)}$ is obtained from derivations ${ \der \num{n_{i}}: \Witypei{h^{i}_{t}}{k^{i}_{t}} }$ ${(1 \leq t\leq q_i)}$, each of depth $0$, by a suitable sequence of applications of rule $(st)$.\\
By Lemma~\ref{lem:subs}, there is a derivation ${\Pi \dem \der \R [ \num{n_{i}} / \x_{i} ]_{i=1}^{p}: \Witypei{h}{k} }$.
By observing the proof of Lemma~\ref{lem:subs}, it is easy to see that such derivation is obtained by replacing axiom ${\y: \Witypei{h^{i}_{t}}{k^{i}_{t}} \der \y: \Witypei{h^{i}_{t}}{k^{i}_{t}} }$ of $\Phi$, for each $\y \in A(\x_{i}, \Phi)$, with derivation ${ \der \num{n_{i}}: \Witypei{h^{i}_{t}}{k^{i}_{t}} }$ of depth $0$, for $1 \leq i \leq p$ and ${1 \leq t\leq q_{i}}$; therefore $\D{\Pi} = \D{\Phi}$, so $\D{\Pi}$ does not depend on the size of the input.\\
Let ${ \M \num{n_{1}} ... \num{n_{p}} \redbetas \R[ \num{n_{i}} / \x_{i} ]_{i=1}^{p} \redbetas \num{\phi(n_{1},...,n_{p})} }$ in $m$ $\beta$-reduction steps;
by Theorem~\ref{th:polystep}, $m\leq |\M \num{n_{1}}...\num{n_{p}} |^{(\D{\Pi}+1)}$ and each intermediate term $\N$ in the reduction sequence is such that ${ |\N| \leq |\M \num{n_{1}}...\num{n_{p}} |^{(\D{\Pi}+1)} }$.
Since a $\beta$-reduction step $\N \redbeta \N'$ can be simulated in time $O(|\N|^2)$ on a Turing machine (see \cite{DBLP:conf/lics/Terui01}), each reduction step takes a time $O( |\M \num{n_{1}}...\num{n_{p}} |^{2(\D{\Pi}+1)}$. Then, since 
%${\num{\phi(n_{1},...,n_{p})} \leq |\M \num{n_{1}}...\num{n_{p}} |^{(\D{\Pi}+1)} }$ and 
the size of the program is a constant with respect to the computation,
the conclusion follows.
\end{proof}

%Therefore, when we consider programs represented by closed $\lambda$-terms in normal form, the application of a program $\M$ to any data $\num{n_{1}},...,\num{n_{p}}$ produces the expected result in a number of steps which is polynomial in the size of the input data.
Since \STR\ is both sound and complete w.r.t. \FPTIME, the following characterization result hold.

\begin{corollary}[\FPTIME\ Characterization]
\STR\ characterizes \FPTIME.
\end{corollary}

\begin{proof}
The proof follows by Lemma~\ref{lem:complSTR} and Theorem~\ref{th:polbound}. 
\end{proof}

Since \STA\ is also sound and complete w.r.t. \FPTIME, the functions representable in \STR\ are exactly the ones representable in \STA.

Consider the \textsl{successors}, concatenating a binary word with either $0$ or $1$:
\begin{itemize}
\item $\underline{\mathtt{succ}_{0}} = \lam \w. \lam f_{0}. \lam f_{1}. \lam \x. \w f_{0} f_{1} (f_{0} \x)$ has type $\Witypei{m}{n} \lin \Witypei{m+1}{n}$ in $\STR$ (resp. $\Wtypei{m}{n} \lin \Wtypei{m+1}{n}$ in $\STA$) and corresponds to the function $f(x)=2x$;
\item $\underline{\mathtt{succ}_{1}} = \lam \w. \lam f_{0}. \lam f_{1}. \lam \x. \w f_{0} f_{1} (f_{1} \x)$ has type $\Witypei{m}{n} \lin \Witypei{m}{n+1}$ in $\STR$ (resp. $\Wtypei{m}{n} \lin \Wtypei{m}{n+1}$ in $\STA$) and corresponds to the function $f(x)=2x + 1$.
\end{itemize}

As an example of the gain in expressivity of $\STR$ with respect to $\STA$, let us consider the \textsl{iteration} of the successor function over binary words; such programming construct is not typable in $\STA$ with a meaningful type, because the polynomial bound is enforced by the fact that functions cannot be iterated, but only composed, in order to forbid the construction of exponential functions.
Nonetheless, in $\STR$ it is possible to type a limited notion of iteration, as shown in the following example.

\begin{example}[Iteration]\label{ex:str-iteration}
In \STR, it is possible to iterate the successor of a binary number a constant number $k$ of times.\\
Consider the $k$-loop term $\underline{\mathtt{ITER}_{k}} = \lam f. \lam \x. f^{k} \x$; in order to be applied to the program $\underline{\mathtt{succ}_{0}}$, the $\underline{\mathtt{ITER}_{k}}$ term can by typed in the following way:
\[ \infer=[(\lin I)]{ \der \lam f. \lam \x. f^{k} \x: \str{ \V}{k} \lin \Witypei{m}{n} \lin \Witypei{m+k}{n} }{ \infer[(m)]{ f: \str{ \V}{k}, \x: \Witypei{m}{n} \der f^{k} \x: \Witypei{m+k}{n} }{ \deduce{ f_{1}: \V_{1}, ... , f_{k}: \V_{k}, \x: \Witypei{m}{n} \der f_{k} ( ... (f_{1} \x ) ... ): \Witypei{m+k}{n} }{ \vdots } } } \]
where $\V_{i} = \Witypei{m+i-1}{n} \lin \Witypei{m+i}{n}$, for $1 \leq i \leq k$.\\
The term $\underline{\mathtt{ITER}_{k}}\ \underline{\mathtt{succ}_{0}}$ can then be obtained by applying rule $(\lin E)$ to the derivation above and to the one obtained as follows:
\[ \infer[(st)]{ \der \underline{\mathtt{succ}_{0}}: \str{ \V}{k} }{ ( \der \underline{\mathtt{succ}_{0}}: \V_{i} )_{1 \leq i \leq k} } \]

Similarly, in order to apply $\underline{\mathtt{ITER}_{k}}$ to $\underline{\mathtt{succ_{1}}}$, the $k$-loop term can be assigned the following type:
\[ \stra{ \Witypei{m}{n} \lin \Witypei{m}{n+1} , ... , \Witypei{m}{n+k-1} \lin \Witypei{m}{n+k} } \lin \Witypei{m}{n} \lin \Witypei{m}{n+k} \]

Observe that this term is not typable in $\STA$, because it is not possible to assign the same type to every $f_{i}$ $(1 \leq i \leq k)$.

\end{example}

\section{Stratification vs Intersection}\label{sec:strat-int}
It is natural to ask if there is a connection between stratified and intersection types. 
Let us consider the set ${\cal I}$ of intersection and quantifier types, where types are strict (no intersection on the right of the arrow) and intersection is a n-ary connective, for $n\geq 2$:
\[
\begin{array}{lcl}
\C & ::= &  \tvar \mid \zeta \rightarrow \C \mid \forall \tvar. \C \\ 
\zeta & ::= & \C \mid \underbrace{\zeta\wedge...\wedge\zeta}_n \mbox{    } (n \geq 2)
\end{array}
\]
where, for simplicity, the type constant $\tvar$ ranges over the same set as in Definition \ref{def:int-types}.
There is a natural translation $(.)^*$ from $\T$ to ${\cal I}$:
\[
\begin{array}{cc}
(\tvar)^*=\tvar,& (\sigma \lin \A)^*=(\sigma)^* \rightarrow (\A)^* \\
 (\stra{\sigma_1,...,\sigma_n})^*=
(\sigma_1)^*\wedge...\wedge(\sigma_n)^*, & (\forall \tvar.\A)^*=\forall\tvar.(\A)^*
\end{array}
\]

The translation can be extended to contexts:
\[ (\emptyset)^* = \emptyset, \qquad(\Gamma, \x:\sigma)^*= (\Gamma)^*, \x:(\sigma)^*\]
Then we can define a type assignment system $\IN$, obtained from $\STR$ by the translation $(.)^*$, such that
for each rule $(R)$ of $\STR$:
$$\infer[(R)]{\Gamma \der \M:\sigma}{(\Gamma_{i} \der \M:\sigma_{i})_{i\in I}}$$
(where the cardinality of $I$ depends on $(R)$) there is a corresponding rule $(R^*)$ in $\IN$:
\[
\infer[(R^*)]{(\Gamma)^* \der \M:(\sigma)^*}{((\Gamma_{i})^* \deri \M:(\sigma_{i})^*)_{i\in I}}
\]
Usually, intersection is considered modulo idempotency ($\zeta=\zeta\wedge\zeta$), commutativity 
($\zeta_1\wedge\zeta_2 = \zeta_2\wedge\zeta_1$) and associativity $ (\zeta_1 \wedge \zeta_2)\wedge \zeta_3=
 \zeta_1 \wedge (\zeta_2\wedge \zeta_3)$. 
It is easy to check that these two systems are equivalent, i.e., $\Gamma \der \M: \sigma$ if and only if $(\Gamma \der \M: \sigma)^*$, with the proviso that intersection is considered modulo idempotency and commutativity, \textsl{but not associativity}.
\section{Conclusion} \label{sec:concl}
We defined a type assignment system, \STR, which not only characterizes all and only the polynomial functions, but is also complete
with respect to strong normalization. The key ingredient for achieving the latter property is the types stratification, i.e. the possibility of contracting different premises $\x_{i}:\sigma_{i}$ ($1 \leq i \leq n$) into a single one $\x:\{\sigma_{1},...,\sigma_{n}\}$. This is clearly inspired by intersection types: indeed, in the previous section we showed that stratification corresponds to non-associative intersection;
in particular, the type ${ \{\sigma_{1},...,\sigma_{n}\} }$ could be written as ${ \sigma_{1}\wedge ...\wedge\sigma_{n} }$, but the first notation seems to better stress the fact that we consider intersection as set formation.
%So our stratification corresponds to use intersection as an idempotent, but non-associative connective, since $\{\{\sigma,\rho\}, \{\tau\}\} \not=\{\sigma, \rho, \tau\}$.
This is not the standard use of intersection, as intersection in the literature usually enjoys idempotency, associativity and commutativity, but all these properties together erase any quantitative information from a typing.
In fact, intersection types have been traditionally used for proving only qualitative properties of terms. 

\medskip
In our setting, if we consider the types of \STR\ without idempotency, in a derivation $\Pi$ the number of premises $\x_{i}:\A_{i}$ where $\x\in \FV{\M}$ corresponds exactly to the number of occurrences of $\x$ in the normal form of the subject of $\Pi$. 
Our first attempt was to design a type assignment similar to \STR\, where types did not enjoy idempotency nor associativity, in order to mimic intersection by multisets instead of sets. Removing associativity is necessary in order to express a bound on the complexity, since the stratification (which derives from the lack of associativity) gives a bound on the number of nested duplications of subterms. This alternative approach is introduced in \cite{DeBen11} and further developed in \cite{DeBenRDR12}. In that system, a property very similar to that of Theorem~\ref{th:polystep} is proved, namely that a term $\M$ typable by a derivation $\Pi$ reduces to normal form by a number of $\beta$-steps bounded by $|\M|^{\D{\Pi} + 1}$, but this result is not very suitable for implicit characterization of complexity classes: indeed non-idempotent types are too informative, 
as there is not a common type that can be assigned to all Church numerals (or binary words) and consequently the notion of data types is not satisfied.  Here we proved that cutting only associativity is sufficient for a polynomial characterization.

\medskip
In the literature, non-idempotent intersection types are used for studying quantitative properties.
Kfoury \cite{kfoury00} connected non-idempotent intersection types with linear $\beta$-reduction, and, together with Wells, he uses non-idempotent intersection for designing a type inference algorithm \cite{KfWells2004}.
Recently non-idempotent intersection types have been used by Pagani and Ronchi Della Rocca for characterizing the
solvability in the resource $\lambda$-calculus \cite{pagani10fossacs}.
In \cite{DiGia08} a game semantics of a typed $\lambda$-calculus has been described in logical form using an intersection type assignment system where the intersection does not enjoy any of its original properties (idempotence, commutativity, associativity). 
Some complexity results have been obtained by De Carvalho \cite{deCarvalho09CORR} and by Bernadet and Lengrand  \cite{bernadetleng11}. A logical description of relational model of $\lambda$-calculus \cite{bucciarelli07csl} has been designed, through a non-idempotent type assigment system, %\cite{paolini12draft},
but in our knowledge this is the first use of intersection types in the ICC setting.

\medskip

%\paragraph{\bf \em Acknowledgments} We would like to thank the referees for their extremely useful observations on the first version of this paper.

\bibliographystyle{plain}
\bibliography{main}

\end{document}